\documentclass[submission,copyright,creativecommons]{eptcs}
\providecommand{\event}{42nd International Conference on Logic Programming (ICLP 2026)} % Name of the event you are submitting to
\usepackage[]{hyperref}

\usepackage{iftex}

\ifpdf
  \usepackage{underscore}         % Only needed if you use pdflatex.
  \usepackage[T1]{fontenc}        % Recommended with pdflatex
\else
  \usepackage{breakurl}           % Not needed if you use pdflatex only.
\fi

\title{Hybrid MKNF with Classical Negation in the Rule Component}

\author{Arun Raveendran Nair Sheela
\institute{Université Clermont Auvergne, LIMOS\\
France}
\institute{Thales}
\email{your.email@example.com}
\and
Christophe Rey
\institute{Université Clermont Auvergne, LIMOS, CNRS\\
France}
\email{christophe.rey@example.com}
\and
Florence De Grancey
\institute{Thales}
\email{florence.degrancey@example.com}
}
\def\titlerunning{Hybrid MKNF with Classical Negation in the Rule Component}
\def\authorrunning{A.R.N. Sheela, C. Rey \&  F. De Grancey}

\usepackage[ruled,vlined,linesnumbered]{algorithm2e}
\SetKwInput{KwRequire}{Require}
\SetKwInput{KwEnsure}{Ensure}
\SetAlgoSkip{0.3em}

\usepackage{enumitem}
\usepackage{tabularray}
\SetAlCapSkip{0.2em}
\SetAlgoInsideSkip{0.2em}
\usepackage{tikz}
\usetikzlibrary{arrows.meta, positioning, calc}
\usepackage{tikz}
\usepackage{amssymb} % Required for \circlearrowright
\usetikzlibrary{arrows.meta, positioning, calc}

\usepackage{amsmath}
\usepackage{booktabs}

\AtBeginEnvironment{gather}{%
  \setlength{\abovedisplayskip}{3pt}%
  \setlength{\belowdisplayskip}{3pt}%
}
\AtBeginEnvironment{equation}{%
  \setlength{\abovedisplayskip}{3.5pt}%
  \setlength{\belowdisplayskip}{3.5pt}%
}

\usepackage{microtype}
\usetikzlibrary{shapes.geometric, arrows.meta, positioning}
\usepackage{booktabs}     % For professional-quality table lines (\toprule, \midrule, \bottomrule)
\usepackage{amsmath}      % For basic mathematical symbols and environments
\usepackage{amssymb}      % For special symbols like \mathcal and \mathbb (used for partitions and truth values)
\usepackage{mathtools}    % Recommended for advanced math formatting

\usepackage{amssymb}

\usepackage[T1]{fontenc}
\usepackage{lmodern}
\usepackage{microtype}

\usepackage{enumitem}
\providecommand{\pred}[1]{\text{\sffamily\slshape #1}}

\usepackage{amsthm}

\usepackage{wasysym}
\usepackage{textcomp}
\newcommand\hybridmknf{\text{hMKNF}\neg}
\newcommand\rK{\mathcal{K}}
\newcommand\rKG{\mathcal{K}_\mathit{G}}
\newcommand\rO{\mathcal{O}}
\newcommand\rP{\mathcal{P}}
\newcommand\rPG{\mathcal{P}_G}
\newcommand\bK{\mathbf{K}}
\newcommand\bnot{\mathbf{not}\,}
\newcommand\EKAKG{\mathsf{EKA}(\rKG)}

\newcommand\OBO[1]{\mathsf{OB}_{\rO,#1}}

\usepackage{graphicx}
\usepackage{multirow}
\usepackage{amsmath}
\usepackage{amsthm}
\makeatletter
\def\thm@space@setup{%
  \thm@preskip=4pt
  \thm@postskip=4pt
}
\makeatother
\newtheorem{theorem}{Theorem}[section]       % Theorem 2.1, 2.2, etc.
\newtheorem{lemma}{Lemma}[section]           % Lemma 2.1, 2.2, etc.
\newtheorem{proposition}{Proposition}[section] % Proposition 2.1, 2.2, etc.
\newtheorem{definition}{Definition}[section]   % Definition 2.1, 2.2, etc.
\newtheorem{example}{Example}[section]       % Example 2.1, 2.2, etc.
\newcommand\rKOP{\mathcal{K = (O,P)}}
\newcommand\rKOPG{\mathcal{K_\mathit{G} = (O,P_\mathit{G})}}

\makeatletter
\usepackage{amsthm}

\begin{document}
\maketitle
\vspace{-1em}
\begin{abstract}
Hybrid MKNF knowledge bases under the well-founded semantics integrate Description Logics with Logic Programming. However, they do not support classical negation in the rule component, limiting their ability to represent explicit negative knowledge. This limitation is particularly significant in safety-critical applications, where reasoning often requires explicit negative information rather than interpreting the absence of information as evidence of absence. To address this issue, we introduce an extension of Hybrid MKNF that supports classical negation in the rule component. We formally define the syntax and semantics of the extended language and present a general procedure for computing its well-founded model.
\end{abstract}

\section{Introduction}
Knowledge representation and reasoning studies formal methods for representing domain knowledge and performing tasks such as query answering through inference. Within this area, the integration of two formalisms, Description Logics (DLs) and Logic Programming, to support both open-world and closed-world reasoning has been widely studied~\cite{ruleontologuessurveyrrecoucilat}. Among existing approaches, hybrid MKNF is particularly expressive, as it combines DL ontologies with logic programming rules under the semantics of Minimal Knowledge and Negation as Failure (MKNF), extending first-order logic with the modal operators $\mathbf{K}$ and $\mathbf{not}$ \cite{MKNF1}. A three-valued characterisation of hybrid MKNF, capturing the
well-founded semantics, was introduced in~\cite{Knorr2011LocalCW}, together with a
bottom-up method for computing the well-founded model in polynomial time, assuming
polynomial-time DL reasoning.

Consider modelling a knowledge base  using hybrid MKNF, where the DL component provides a static description of domain entities and relations, and the rule component captures dynamic aspects such as events, actions, and state changes.   The approach of~\cite{Knorr2011LocalCW} does not support classical negation in the rule component, despite it being allowed by MKNF semantics. Instead, negative information is expressed only via default negation (\textbf{not}), allowing the conclusion of \textbf{not} $A$ when $A$ is not derivable. Consequently, negated conclusions rely on the absence of knowledge rather than explicit evidence, which may lead the rule component to infer domain states that do not reflect the actual situation. We illustrate this limitation with a simple example concerning the operational status of an airport runway: \emph{A runway is open if it is associated with an airport
and is not known to be closed or obstructed. } This requirement can be formalized as a hybrid MKNF knowledge base
$\mathcal{K}=(\mathcal{O},\mathcal{P})$, where the ontology $\mathcal{O}$ contains
the axiom
\(
\mathsf{OpnRwy} \sqsubseteq \mathsf{Rwy},
\)
and the rule component $\mathcal{P}$ includes the following rule:
\begin{gather}\label{rule:statechaopnern} \small
\mathbf{K}\mathsf{OpnRwy}(Y) \leftarrow\ 
\mathbf{K}\mathsf{Airport}(A),\ 
\mathbf{K}\mathsf{hasRwy}(A,Y),\ 
\mathbf{not}\,\mathsf{ob}(A,Y),\ 
\mathbf{not}\,\mathsf{CldRwy}(Y).
\end{gather}

Inferring runway availability from the absence of obstacle information (\(\mathbf{not}\,\mathsf{ob}(A,Y)\)) is unsafe, since missing or delayed data does not imply that no obstacle exists. Rule~\ref{rule:statechaopnern} formalizes an operational condition that may vary dynamically at runtime.
Representing such dynamic operational conditions in DL is inappropriate, and the
default assumption $\mathbf{not}\,\mathsf{CldRwy}(Y)$ cannot be expressed using DL axioms alone.  Moreover, introducing auxiliary confirmation predicates does not resolve the issue, as confirmations of presence and absence may coexist without yielding a contradiction. A more robust representation relies on classical negation: 
\begin{gather}\label{rule:classical-neg}\small
\mathbf{K}\mathsf{OpnRwy}(Y) \leftarrow\ 
\mathbf{K}\mathsf{Airport}(A),\ 
\mathbf{K}\mathsf{hasRwy}(A,Y),\ 
\mathbf{K}\neg\mathsf{ob}(A,Y),\ 
\mathbf{not}\,\mathsf{CldRwy}(Y).
\end{gather}
The rule \eqref{rule:classical-neg} can be read as: \emph{A runway is open if
an airport has the runway, it is not known to be closed, and it is known to have no
obstacles.}
The use of classical negation makes contradictions explicit at the semantic level.  Additional applications of hybrid MKNF extended with classical negation are discussed in \cite{EPTCS439.24}.

 Classical negation can be incorporated
indirectly in the rule component via a syntactic transformation for unary predicates, as proposed
in~\cite{Knorr2011LocalCW}. In this transformation, for each unary predicate $A$,
a fresh concept $B$ is introduced in the DL component together with the axiom
$B \equiv \neg A$, allowing $B$ to represent the classical negation of $A$ in
rules. This approach requires additional knowledge engineering effort and is
restricted to unary predicates. A further limitation of the approach in~\cite{Knorr2011LocalCW} is that it correctly
computes the well-founded model only for \emph{coherent hybrid MKNF}
knowledge bases, a  subclass of hybrid MKNF \cite{LIU2017123}. For non-coherent knowledge bases, the bottom-up computation of~\cite{Knorr2011LocalCW} produces a well-founded model that differs from the declarative semantics, as shown in~\cite{LIU2017123}.  To address this,~\cite{LIU2017123}
propose a general guess-and-check method for computing three-valued MKNF models.
However, this approach is inherently nondeterministic and does not formally
establish the treatment of classical negation in the rule component.

In this work, we address this gap by extending hybrid MKNF knowledge bases to unrestricted usage of  classical negation in the rule component. We formally establish the treatment of classical negation by adapting key intuitions from \cite{Knorr2011LocalCW,LIU2017123,ji2017wellfoundedoperatorsnormalhybrid}. Our main contributions are as follows:
\vspace{-.5em}
\begin{itemize}[itemsep=0pt]
   \item[(i)] We introduce \(\hybridmknf\), an extension of hybrid MKNF under the well-founded semantics that supports classical negation in the rule component (Section~\ref{sec:semantics}).
  \item[(ii)] We present a three-valued characterization of $\hybridmknf$ knowledge bases via \emph{stable partitions} adapted from~\cite{LIU2017123}, where the \emph{well-founded partition} is the unique minimal stable partition which yields  the \emph{well-founded model} (Section~\ref{sec:three-valued}).
  \item[(iii)] We present an operational procedure for the computation of the well-founded partition of $\hybridmknf$ knowledge bases (Section~\ref{sec:computaiton}).
\end{itemize}
\subsection{Three-valued semantics for MKNF}
The Logic of MKNF is an extension of  first-order logic  with two modal operators, $ \mathbf{K}$ and $\mathit{ \mathbf{not}}$ \cite{MKNF1}. Intuitively, the operator $\mathbf{K}$ represents explicitly known information, while $\mathbf{not}$ denotes information that is not explicitly known in a knowledge base. Together, these two operators behave in a way that corresponds to the closed-world assumption.

Let $\Sigma = (\Sigma_c, \Sigma_f, \Sigma_p)$ be a first-order signature, where $\Sigma_c$, $\Sigma_f$, and $\Sigma_p$ are the sets of constants, function symbols, and predicate symbols, respectively, with $\Sigma_p$ containing the equality predicate $\approx$.
A first-order atom $P(t_1,\dots,t_n)$ is an MKNF formula, where $P$ is a predicate symbol and $t_1,\dots,t_n$ are first-order terms. If $\varphi$ and $\varphi'$ are MKNF formulae, then $\neg \varphi$, $\exists x:\varphi$, $\mathbf{K}\varphi$, $\mathbf{not}\,\varphi$, and $\varphi \wedge \varphi'$ are also MKNF formulae. Moreover, formulae constructed using the connectives $\vee$, $\supset$, and $\equiv$, as well as the universal quantifier $\forall$, are interpreted as in first-order logic. A MKNF formula \(\varphi\) is ground if it contains no variables.  A first-order interpretation  \(I\) over a signature \(\Sigma\) with domain \(\Delta\)
maps each constant \(a \in \Sigma_c\) to an element \(a^I \in \Delta\),
each \(n\)-ary function symbol \(f \in \Sigma_f\) to a function
\(f^I : \Delta^n \to \Delta\), and each \(n\)-ary predicate symbol
\(P \in \Sigma_p\) to a relation \(P^I \subseteq \Delta^n\).
Furthermore, for every \(\alpha \in \Delta\), there exists a designated
constant \(n_\alpha \in \Sigma\), called a name, such that \(n_\alpha^I = \alpha\).

In \cite{Knorr2011LocalCW}, a three-valued MKNF structure is defined as a tuple
$(I,\mathcal{M},\mathcal{N })$, where $I$ is a first-order interpretation and
$\mathcal{M}=\langle M,M_1\rangle$ and $\mathcal{N}=\langle N,N_1\rangle$
are pairs of sets of first-order interpretations such that
$M_1 \subseteq M$ and $N_1 \subseteq N$. The evaluation of an MKNF formula $\varphi$ under the truth values
$\{\mathbf{t},\mathbf{u},\mathbf{f}\}$, ordered by $\mathbf{f}<\mathbf{u}<\mathbf{t}$,
with respect to a three-valued MKNF structure, is defined in Fig. \ref{eq:mknfevalaution}:\vspace{-1em}
\begin{figure}[htp] \label{eq:mknfevalaution}\footnotesize
 \begin{gather}
\mathcal{(I,M,N)}\bigl(P(t_1,\dots,t_n)\bigr) =
\begin{cases}
\mathbf{t} & \text{iff } (t_1^I,\dots,t_n^I) \in P^I,\\[0.5ex]
\mathbf{f} & \text{iff } (t_1^I,\dots,t_n^I) \notin P^I.
\end{cases} \\
\mathcal{((I,M,N)}(\neg \varphi) = 
\begin{cases} 
\mathbf{t} & \text{if } \mathcal{(I,M,N)}(\varphi) = \mathbf{f}, \\
\mathbf{u} & \text{if } \mathcal{(I,M,N)}(\varphi) = \mathbf{u}, \\
\mathbf{f} & \text{if } \mathcal{(I,M,N)}(\varphi) = \mathbf{t},
\end{cases} \\
\mathcal{(I,M,N)}(\varphi_1 \land \varphi_2) = \min \{ \mathcal{(I,M,N)}(\varphi_1), \mathcal{(I,M,N)}(\varphi_2) \}, \\
\mathcal{(I,M,N)}(\varphi_1 \supset \varphi_2) = 
\begin{cases}
\mathbf{t} & \text{if } \mathcal{((I,M,N)}(\varphi_2) \ge \mathcal{(I,M,N)}(\varphi_1), \\
\mathbf{f} & \text{otherwise},
\end{cases} \\
\mathcal{(I,M,N)}(\exists x: \varphi) = \max \{ \mathcal{(I,M,N)}(\varphi[n_\alpha/x]) \mid \alpha \in \Delta \},  \\
 \mathcal{(I,M,N)}( \mathbf{K} \  \varphi) = 
\begin{cases}
\mathbf{t} & \text{if } (\mathcal{J}, \langle M, M_1 \rangle, \mathcal{N})(\varphi) = \mathbf{t} \text{ for all } \mathcal{J} \in M, \\
\mathbf{f} & \text{if } (\mathcal{J}, \langle M, M_1 \rangle, \mathcal{N})(\varphi) = \mathbf{f} \text{ for some } \mathcal{J} \in M_1, \\
\mathbf{u} & \text{otherwise},
\end{cases} \\
\mathcal{(I,M,N)}(\textbf{not } \varphi) = 
\begin{cases}
\mathbf{t} & \text{if } (\mathcal{J}, \mathcal{M}, \langle N, N_1 \rangle)(\varphi) = \mathbf{f} \text{ for some } \mathcal{J} \in N_1, \\
\mathbf{f} & \text{if } (\mathcal{J}, \mathcal{M}, \langle N, N_1 \rangle)(\varphi) = \mathbf{t} \text{ for all } \mathcal{J} \in N, \\
\mathbf{u} & \text{otherwise}.
\end{cases}
\end{gather}
\vspace{-1em}
\caption{Evaluation of MKNF formulae with respect to a MKNF structure}
\vspace{-1.5em}
\end{figure}

\normalsize
A three-valued MKNF interpretation is a pair $(M, N)$ consisting of two sets of interpretations $M$ and $N$ such that $\emptyset \subset N \subseteq M$. It satisfies a closed MKNF formula $\varphi$, written as $(M, N) \models \varphi$, if and only if $(I, \langle M, N \rangle, \langle M, N \rangle)(\varphi) = \mathbf{t}$. If $M = N$, then the interpretation pair is called total.  An MKNF interpretation pair $(M, N)$ is a \emph{three-valued MKNF model} for a closed MKNF formula $\varphi$ if
$(i)$ \emph{Satisfaction:} $(M, N)$ satisfies $\varphi$, and $(ii)$\emph{Maximality:}  for every MKNF interpretation pair $(M^\prime, N^\prime)$ with $M \subseteq M^\prime$ and $N \subseteq N^\prime$, where at least one inclusion is proper and $M^\prime = N^\prime$ if $M = N$, there exists $I^\prime \in M^\prime$ such that  \(
    (I^\prime, \langle M^\prime, N^\prime \rangle, \langle M, N \rangle)(\varphi) \neq \mathbf{t}.
    \)  
A closed MKNF formula $\varphi$ is said to be \emph{MKNF-consistent} if there exists a three-valued MKNF model $(M, N)$ for $\varphi$,  Otherwise, $\varphi$ is \emph{MKNF-inconsistent}. Under MKNF semantics, a knowledge base may admit multiple three-valued MKNF models.
When it exists, the \emph{well-founded model} is the unique model that maximizes
the set of undefined modal atoms among all three-valued models.
Formally, let $(M,N)$ denote the well-founded model.
Then, for any other three-valued MKNF model $(M_1, N_1)$, it holds that
$M_1 \subseteq M$ and $N \subseteq N_1$.
\subsection{Hybrid MKNF}
A hybrid MKNF knowledge base is a pair $\rKOP$ consisting of a decidable DL knowledge base O, which must be translatable into function-free first-order logic with equality and guarantee decidability of satisfiability and instance checking, and a finite set P of MKNF rules of the form \cite{hybridMknfintroduction}:
\begin{equation} \label{rule:extended}
\mathbf{K} \ H  \gets \mathbf{K} \ A_1, \dots, \mathbf{K} \ A_n, \ \mathbf{not} \ B_1, \dots,  \mathbf{not \ } B_m.
\end{equation}
\normalsize
where $H$, $A_i$, and $B_j$ are function-free first-order atoms. A rule $r$ is called positive if $m=0$, and it is called a fact if
$n=m=0$. An atom occurring in $\mathcal{P}$ is called a DL-atom if its predicate symbol
occurs in the ontology $\mathcal{O}$; otherwise, it is called a non-DL atom.
A rule is called ground if it contains no variables.

Let $\pi$ denote the MKNF translation of a hybrid MKNF knowledge base.
The translation $\pi(\mathcal{O})$ maps the DL ontology into first-order formulae.
For a rule $r\in\mathcal{P}$ with free variables $\mathrm{\vec{x}} $, the translation is
defined as:
\begin{gather}\scriptsize
     \mathcal{\pi(\textit{r})} = \forall \mathrm{\vec{x}} : (\mathbf{K} \  A_1 \wedge \dots \wedge \mathbf{K} \ A_m \wedge \mathbf{not} \, B_1 \wedge \dots \wedge \mathbf{not } \, B_n \supset \mathbf{K} \ H  ) \label{rule:ruletranslation} \\
       \mathcal{ \pi(P) = \bigwedge_{\textit{r} \in P} \pi(\textit{r})} \qquad  \pi(\rK) = \mathbf{K} \ \pi(\rO) \wedge \pi(\rP)\label{rule:mknftranslation}
\end{gather}
A hybrid MKNF knowledge base $\mathcal{K}$ is MKNF-consistent if its MKNF translation $\pi(\mathcal{K})$ admits a three-valued MKNF model $(M,N)$; otherwise, it is MKNF-inconsistent.  An MKNF rule is said to be DL-safe if every variable occurring in the rule
also occurs in at least one non-DL atom in its body.
A hybrid MKNF knowledge base $\mathcal{K}$ is DL-safe if all its MKNF rules
are DL-safe. Decidability of $\mathcal{K}$ follows from the decidability of the underlying DL,
the restriction that all atoms occurring in MKNF rules are function-free first-order
atoms, and DL-safety. The original MKNF semantics may yield counterintuitive results when combining DLs
and logic programs due to arbitrary universes and varying interpretations of
constants; this issue is addressed in~\cite{mknf+} by adopting the
\emph{standard name assumption},  where interpretations are Herbrand ones with a
countably infinite number of additional constants.

\section{Extension with Classical Negation}\label{sec:semantics}
We extend hybrid MKNF knowledge bases by permitting classical negation in the rule
component.
The resulting formalism, denoted $\hybridmknf$, preserves the same restrictions as hybrid MKNF knowledge bases, including DL-safety \cite{mknf+} to ensure decidability, while allowing both positive and classically negated modal atoms in the rule component. We use the shorthand $h(r) \gets b^+(r), b^-(r)$ to denote a rule of the form~\eqref{rule:extended}.
\begin{definition}
  A \(\hybridmknf\) knowledge base is a pair  \(\rKOP\),  where $\mathcal{O}$ is a DL ontology and $\mathcal{P}$ is a finite set of MKNF rules.  Each rule $r \in \mathcal{P}$ is of the form,  \(h(r) \gets b^+(r),  b^-(r).\) where $h(r)=\mathbf{K}\,H$,
$b^{+}(r)=\{\mathbf{K}\,A_1,\dots,\mathbf{K}\,A_n\}$, and 
$b^{-}(r)=\{\mathbf{not}\,B_1,\dots,\mathbf{not}\,B_m\}$. Moreover, $\mathbf{K}\,\xi \in h(r) \cup b^{+}(r)$ and $\mathbf{not}\,\xi \in b^{-}(r)$, where $\xi$ is either a first-order atom $A$ or its classical negation $\neg A$. Also, denote
\(
\mathbf{K}(b^{-}(r)) = \{\mathbf{K}\,a \mid \mathbf{not}\,a \in b^{-}(r)\}
\) and \(b(r) = b^+(r) \cup  b^-(r)\).
\end{definition}
As defined in~\eqref{rule:ruletranslation} and~\eqref{rule:mknftranslation}, a
$\hybridmknf$ knowledge base is translated into the MKNF formula
$\pi(\mathcal{K}) = \mathbf{K},\pi(\mathcal{O}) \wedge \pi(\mathcal{P})$.
The component $\pi(\mathcal{P})$ may contain modal atoms $\mathbf{K}A$ and
$\mathbf{not} \ A$, as well as their classically negated forms
$\mathbf{K}\neg A$ and $\mathbf{not} \ \neg A$.
Since classical negation is part of the MKNF language, these atoms are
interpreted under the standard MKNF semantics.
 Let $\mathcal{(I, M, N)}$ be
a three-valued MKNF structure and $\varphi$ a formula. According to the evaluation
in Fig.~\ref{eq:mknfevalaution}, classically negated modal atoms are
evaluated with respect to $\mathcal{(I, M, N)}$:
\vspace{-1em}

\scriptsize{
\begin{align}
\mathcal{(I, M, N)}(\mathbf{K} \neg \varphi) &=
\begin{cases}
\mathbf{t} & \text{if } (\mathcal{J}, \langle M, M_1 \rangle, \mathcal{N})(\varphi) = \mathbf{f} \text{ for all } \mathcal{J} \in M, \\
\mathbf{f} & \text{if } (\mathcal{J}, \langle M, M_1 \rangle, \mathcal{N})(\varphi) = \mathbf{t} \text{ for some } \mathcal{J} \in M_1, \\
\mathbf{u} & \text{otherwise,}
\end{cases} \label{eq:truth-Kneg}\\
\mathcal{(I, M, N)}(\textbf{not } \neg \varphi) &=
\begin{cases}
\mathbf{t} & \text{if } (\mathcal{J}, \mathcal{M}, \langle N, N_1 \rangle)(\varphi) = \mathbf{t} \text{ for some } \mathcal{J}\in N_1, \\
\mathbf{f} & \text{if } (\mathcal{J}, \mathcal{M}, \langle N, N_1 \rangle)(\varphi) = \mathbf{f} \text{ for all } \mathcal{J} \in N, \\
\mathbf{u} & \text{otherwise,} \label{eq:truth-notneg}
\end{cases}
\end{align} 
}
\normalsize
With respect to $\langle M, M_1 \rangle$, a classically negated modal atom
$\mathbf{K}\neg \varphi$ is \emph{true} if $\varphi$ is \emph{false} in all interpretations of
$M$, \emph{false} if $\varphi$ is  \emph{true}   in some interpretation of $M_1$, and
\emph{undefined} otherwise.
The evaluation of $\mathbf{not}$-atoms is defined symmetrically to that of
$\mathbf{K}$-atoms.
As a consequence, MKNF semantics enforces a \emph{coherence principle} by default:
if $\mathbf{K}\neg A$ holds in a model, then $\mathbf{not}\,A$ also holds; similarly,
if $\mathbf{K}A$ holds in a model, then $\mathbf{not}\,\neg A$ holds.

\section{Three-valued Model for \(\hybridmknf\) }\label{sec:three-valued}
This section develops the formal basis for constructing three-valued MKNF models
for $\hybridmknf$ knowledge bases, building on partial partitions
from~\cite{Knorr2011LocalCW} and stable partitions
from~\cite{LIU2017123}, which are shown to extend to $\hybridmknf$.

Given a 
 \(\hybridmknf\)  knowledge base $\mathcal{K}=(\mathcal{O},\mathcal{P})$, its ground instantiation is $\mathcal{K}_G=(\mathcal{O},\mathcal{P}_G)$, where $ \mathcal{P}_G
=
\{\, r\theta
\mid
r \in \mathcal{P},\;
\theta \text{ maps each variable occurring in } r
\text{ to a constant of } \rKG
\,\}
$.
$\mathcal{K}$ and $\mathcal{K}_G$ have the same models \cite{Knorr2011LocalCW, mknf+}, and we henceforth consider only ground 
 \(\hybridmknf\) knowledge bases. Following the approach of \cite{Knorr2011LocalCW}, we construct a finite set of
ground modal atoms from a ground \(\hybridmknf\) knowledge base
$\rKG=(\rO,\rPG)$, analogous to a Herbrand base. 
\begin{definition}\label{def:ekag}
Let $\mathcal{K}_G=(\mathcal{O},\mathcal{P}_G)$ be a ground \(\hybridmknf\) knowledge base.
The set of all modal atoms of $\mathcal{K}_G$, denoted by $\EKAKG$, is defined as
$\EKAKG=\{\bK\xi \mid r\in\mathcal{P}_G,\ \bK\xi\in \pred{h}(r)\cup \mathrm{b}^+(r)\cup \bK(\mathrm{b}^-(r)),\ \xi \text{ is } A \text{ or } \neg A\}$.
A partial partition of $\EKAKG$ is a pair $(T,F)$ with $T,F\subseteq\EKAKG$ and
$T\cap F=\emptyset$.
The set $U$ is given by $U=\EKAKG\setminus (T\cup F)$.
\end{definition}

The relationship between a partial partition $(T,F)$ of $\EKAKG$ and a three-valued MKNF interpretation is formalized in the definition below.
\begin{definition}\label{extendedKAthreeinducedparition}
Let $\rKOPG$ be a ground \(\hybridmknf\) knowledge base.  The partial partition $(T,F)$ of $\EKAKG$ is induced by a three-valued interpretation $(M,N)$  if, for every  modal atom \(\mathbf{K}L \in \EKAKG\)  where  $L$ is either $\xi$ or $\neg \xi$, the following conditions hold: $(i)$  \(\mathbf{K} L \in T \text{ implies }\forall I \in M: (I, \langle M, N \rangle, \langle M, N \rangle)(\mathbf{K} L) = \mathbf{t}\); $(ii)$ $\mathbf{K} L \in F\text{ implies }\forall I \in M: (I, \langle M, N \rangle, \langle M, N \rangle)(\mathbf{K} L) = \mathbf{f}$ $(iii)$ $\mathbf{K} L \notin T \cup F$ $\text{ implies } \forall I \in M: (I, \langle M, N \rangle, \langle M, N \rangle)(\mathbf{K} L) = \mathbf{u}$.
\end{definition}

We now define the notion of objective knowledge, adapted from \cite{Knorr2011LocalCW}.
 \begin{definition}\label{definitiobejectiveknowled} Let $\rKOPG$ be a ground \(\hybridmknf\) knowledge base, and let $S \subseteq \EKAKG$. The \emph{objective knowledge} of $S$ with respect tot  $\mathrm{\rKG}$ is \(\OBO{S} = \{\pi(\mathcal{O})\} \cup \{\xi \mid \mathbf{K} \xi \in S\} \cup \{\neg \xi \mid \mathbf{K} \neg \xi \in S\}.\)
\end{definition}

A three-valued MKNF interpretation $(M,N)$ with respect to a partial partition
$(T,F)$ of $\EKAKG$ is defined as follows.
The set $M$ contains the models of $\pi(\mathcal{O})$ extended with $T$
(Def.~\ref{definitiobejectiveknowled}),
while $N$ contains the models extended with $T$ and $U$
(Prop.~3 in~\cite{Knorr2011LocalCW}). The intuition behind this construction is to obtain a finite representation of a
$\hybridmknf$ knowledge base that captures its semantics.

\begin{proposition}\label{propos:stablepartitionstheorem}
Let $(M,N)$ be a three-valued MKNF model of a ground $\hybridmknf$ knowledge base $\rKG$, 
and let $(T,F)$ be the partial partition induced by $(M,N)$.
Then
\(
M = \{\, I \mid I \models \OBO{T} \,\}\) and 
\(
N = \{\, I \mid I \models \OBO{\EKAKG \setminus F} \,\}.
\)
\end{proposition}

\begin{proof}
Proof is provided in Prop. \ref{propos:stablepartitionstheoremappneda1}.
\end{proof}

\begin{definition}\label{def:indieuceMKNFinterpreation}
Let $\rKOPG$ be a ground \(\hybridmknf\) knowledge base. Let $(T,F)$ be a partial partition of $\EKAKG$. The MKNF interpretation pair $(M,N)$ induced by $(T,F)$ is defined as follows:
\(
M = \{\, I \mid I \models \OBO{T} \,\}\) and 
\(
N = \{\, I \mid I \models \OBO{\EKAKG \setminus F} \,\}.
\)
\end{definition}

\subsection{Stable Partition}
Following~\cite{LIU2017123}, a stable partition is a partial partition $(T,F)$ that
induces, via Def.~\ref{def:indieuceMKNFinterpreation}, a three-valued MKNF
interpretation $(M,N)$ that is a model of a $\hybridmknf$ knowledge base. 
By construction, $(M,N)$ satisfies the DL component of $\rKG$.
Hence, $(T,F)$ is stable iff (i) $(M,N)$ satisfies the rule component of $\rKG$, and
(ii) $(M,N)$ is \emph{maximal}, that is, there is no $(M',N')$ with
$M \subset M'$ or $N \subset N'$, still satisfying the rule component; moreover,
if $M = N$, then $M' = N'$.

The rule component is evaluated with respect to $(M,N)$ by assigning each modal
atom $A$ the truth value $A[T,F]$, which denotes the value of $A$ under the
three-valued MKNF interpretation $(M,N)$ induced by $(T,F)$.
This evaluation relies on $M$ and $N$ being sets of interpretations satisfying
$\OBO{T}$ and $\OBO{\EKAKG \setminus F}$, respectively.

\begin{lemma}\label{lem:3valued-eval}
Let $\rKOPG$ be a ground $\hybridmknf$ knowledge base and $(T,F)$ a partial
partition of $\EKAKG$.   For every $\xi$ such that
$\mathbf{K}\xi \in \EKAKG$:
{\footnotesize
\[
\begin{array}{@{}c@{\qquad}c@{}}
\mathbf{K}\xi[T,F]=
\begin{cases}
\mathbf{t} & \OBO{T}\models\xi,\\
\mathbf{f} & \OBO{\EKAKG\setminus F}\not\models\xi,\\
\mathbf{u} & \text{otherwise}.
\end{cases}
&
\mathbf{not}\ \xi[T,F]=
\begin{cases}
\mathbf{t} & \OBO{\EKAKG\setminus F}\not\models\xi,\\
\mathbf{f} & \OBO{T}\models\xi,\\
\mathbf{u} & \text{otherwise}.
\end{cases}
\end{array}
\]
}

\end{lemma}
Intuitively, Lem.~\ref{lem:3valued-eval}, based on
Prop.~\ref{propos:stablepartitionstheorem}, shows that for any \(x \in \mathbb{T}\), \(\mathbf{K}\xi[T,F] = x\)
(resp.\ \(\mathbf{not}\ \xi[T,F] = x\)) holds if and only if
\(\mathbf{K}\xi\) (resp.\ \(\mathbf{not}\ \xi\)) evaluates to \(x\)  with respect to MKNF interpretation $(M,N)$ induced by  partial partition $(T,F)$.

\begin{lemma}\label{lemma:proofstbaleparitiomn2}
Let $(T,F)$ be a partial partition of $\EKAKG$ and let $(M,N)$ be the MKNF interpretation
induced by $(T,F)$. Let
\(\mathbb{T}=\{\mathbf{f},\mathbf{u},\mathbf{t}\}\) with
\(\mathbf{f}<\mathbf{u}<\mathbf{t}\). Then, for every rule $r\in\rPG$,
\(
\big( \forall I \in M,\ (I,\langle M,N\rangle,\langle M,N\rangle)(\pi(r))=\mathbf{t} \big)
\ \text{iff}\ 
\big( h(r)[T,F]\ge b(r)[T,F] \big),
\)
where
\(
b(r)[T,F]=\min\{\ell[T,F]\mid \ell\in b^+(r)\cup b^-(r)\}.
\)
\end{lemma}
\vspace{-.3em}
Lem.~\ref{lem:3valued-eval} and Lem.~\ref{lemma:proofstbaleparitiomn2} establish how the
rule component of a $\hybridmknf$  is evaluated with respect to a
partial partition, allowing us to define the  stable partitions.

\begin{definition}[Stable partition]\label{def:semanticsofstablepartition}
Let $\rKOPG$ be a ground $\hybridmknf$ knowledge base and $(T,F)$ a partial
partition of $\EKAKG$.
We call $(T,F)$ a \emph{stable partition} if and only if: $\textnormal{(i)}$  $\OBO{\EKAKG\setminus F}$ is satisfiable;
\begin{enumerate}[leftmargin=*,itemsep=1.5pt]
\item[\textnormal{(ii)}]
\textnormal{(ii.1)} for every $\mathbf{K}\xi\in\EKAKG$,
$\OBO{T}\models\xi$ implies $\mathbf{K}\xi\in T$ and
$\OBO{\EKAKG\setminus F}\not\models\xi$ implies $\mathbf{K}\xi\in F$;
\textnormal{(ii.2)} $\forall I \in M: (I,\langle M,N\rangle ,\langle M,N \rangle )(\pi(r))=t$ for all $r\in \rPG$ and $I\in M$, where $(M,N)$ is
the MKNF interpretation induced by $(T,F)$;
\item[\textnormal{(iii)}] for any $(T',F')$ with $T'\subseteq T$ and $F \subseteq F^\prime$
(at least one proper), either there exists $\mathbf{K}\xi\in\EKAKG\setminus T'$ such that
$\OBO{T'}\models\xi$ 
or there exists $\mathbf{K}\xi\in F^\prime \setminus F$ such that
$\OBO{\EKAKG\setminus F'}\models\xi$, or there exists $r\in \rPG$ such that
$\forall I \in M'$, $(I,\langle M',N'\rangle ,\langle M,N\rangle)(\pi(r))=f$, where $(M',N')$ is induced by $(T',F')$.
\end{enumerate}
\end{definition}

\begin{theorem}\label{thereoe:stbaleparitito}
Let $\rKOPG$ be a ground $\hybridmknf$ knowledge base. Let $(T,F)$ be a partial partition of $\EKAKG$.
$(T,F)$ is stable
if and only if the three-valued MKNF interpretation $(M,N)$ induced by $(T,F)$
is a three-valued MKNF model of $\rKG$.
\end{theorem}
\begin{proof}
   Proof is provided in Thm. \ref{thereoe:stbaleparititoa3}. 
\end{proof}
Having established the semantics of stable partitions in Def.~\ref{def:semanticsofstablepartition}, we now present a computation procedure,  adapted from~\cite{LIU2017123}, for checking whether a partial partition is stable. This procedure is used during the computation of the well-founded model of a $\hybridmknf$ knowledge base explained in Section \ref{sec:computaiton}. 

To simplify computation, $\mathbf{K} \neg A$ or $\mathbf{not \ } \neg A$ in the rule component are handled by treating  $\neg A$ as positive atom. The objective knowledge $\OBO{S}$ is constructed by translating
$\mathbf{K}A \in S$ as $A$ and $\mathbf{K}\neg A \in S$ as $\neg A$, with classical
negation treated as in first-order logic (Def.~\ref{definitiobejectiveknowled}). For any
first-order atom $A$, its complement is denoted by $\overline{A}$, where
$\overline{A}=\neg A$ and $\overline{\neg A}=A$.

\begin{definition}\label{def:monotoncoepratos}
    Let $\rKOPG$ be a ground \(\hybridmknf\) knowledge base and let
$(T,F)$ be a partial partition of $\EKAKG$ and \(X,S \subseteq \EKAKG\).
We define the following operators:
{\small
\begin{align}
T^S_{\mathcal K_G}(X)
= {} &
\{\, h(r) \mid r\in\mathcal P_G,\;
b^+(r)\subseteq  X,\;
\mathbf K(b^-(r))\subseteq S \,\}
\nonumber\\
&\cup
\{\, \mathbf K\xi \mid \mathbf K\xi\in\EKAKG,\;
\OBO{X}  \models\xi \,\},
\label{eq:true}
\\\small
TU^S_{\mathcal K_G}(X)
= {} &
\{\, h(r) \mid r\in\mathcal P_G,\;
b^+(r)\subseteq  X,
\mathbf K(b^-(r))\subseteq S  \; \text{ and } \OBO{X} \not\models \overline{h(r)}\,\}
\nonumber\\
&\cup
\{\, \mathbf K\xi \mid \mathbf K\xi\in\EKAKG,\;
\OBO{X}  \models\xi \,\},
\label{eq:true1}
\end{align}
}
\normalsize
\end{definition}

The operators $T^S_{\mathcal K_G}$ and $TU^S_{\mathcal K_G}$ are monotonic (\cite{Knorr2011LocalCW}).
As they range over the finite lattice of partial partitions of
$\EKAKG$, they reach a least fixpoint in finitely many iterations, starting from the empty set.
We denote $\Gamma(S)$ as the least fixpoint of $T^S_{\mathcal K_G}$, $\Gamma'(S)$ as the least fixpoint of $TU^S_{\mathcal K_G}$ and $Fa(S)$ as $\EKAKG \setminus \Gamma'(S)$.
\begin{proposition}\label{propos:computingstablepartitions}
Let $\rKOPG$ be a ground $\hybridmknf$ knowledge base.
$(T,F)$ is a stable partition of $\EKAKG$ if and only if (i) $T =  \Gamma(F) $, (ii) $F = Fa(\EKAKG \setminus T)$ and (iii) $\OBO{\Gamma(\EKAKG \setminus T)}$ is satisfiable.
\end{proposition}
\begin{proof}
Proof is provided in Prop. \ref{propos:computingstablepartitionss3}.
\end{proof}

\vspace{-1.5em}
\paragraph{\textbf{Well-founded partition.}}
A $\hybridmknf$ knowledge base may admit multiple three-valued MKNF models, each
inducing a stable partition (Definition~\ref{def:semanticsofstablepartition}).
\begin{definition}[Well-founded partition]\label{def:well-foundedpartition}
Let $\rKG$ be a ground $\hybridmknf$ knowledge base.
A stable partition $(T_W,F_W)$ of $\rKG$ is called the \emph{well-founded partition} of $\rKG$
if, for every stable partition $(T,F)$ of $\rKG$, it holds that
$T_W \subseteq T$ and $F_W \subseteq F$.
\end{definition}

\begin{theorem}[Well-founded Model]\label{theorem:wfmodel}
Let $\rKG$ be a ground $\hybridmknf$ knowledge base.
If $(T_W,F_W)$ is the well-founded partition of $\rKG$,
then the MKNF interpretation pair $(M_W,N_W)$ induced by $(T_W,F_W)$
is the well-founded model of $\rKG$.
\end{theorem}
\section{Computation of Well-founded Partition}\label{sec:computaiton}
The computation of the well-founded partition for a $\hybridmknf$ knowledge base follows a three-phase strategy (Fig.~\ref{fig:workflow}). The first phase applies a well-founded operator adapted from \cite{ji2017wellfoundedoperatorsnormalhybrid} (Subsection~\ref{sec:well-foundedoperator}), followed by a unit propagation phase (Subsection~\ref{sec:unitpropagation}), and finally a guess-and-check approach (Subsection~\ref{sec:guessandcheck}). This strategy prioritises deterministic computation by relying on the first two phases and resorting to guess-and-check only when necessary, since not all $\hybridmknf$ knowledge bases yield a well-founded partition after each phase. Some illustrative examples are  moved to \ref{sec:illustrativeexamples}.

\begin{figure}[ht] \label{fig:workflow} \scriptsize
\begin{flushleft}
\begin{tikzpicture}[
    scale=0.74, transform shape,
    node distance=1.0cm and 0.8cm,
    phase/.style={rectangle, draw, minimum width=2.8cm, minimum height=1.2cm, align=center, font=\small\sffamily, line width=0.8pt},
    result/.style={rectangle, draw, double, fill=gray!5, minimum width=3.5cm, minimum height=0.7cm, align=center, font=\small\bfseries\sffamily, line width=1pt},
    decision/.style={diamond, draw, aspect=2.0, align=center, font=\tiny\sffamily, inner sep=1.5pt, line width=0.8pt},
    arrow/.style={->, >=Stealth, line width=0.7pt},
    label/.style={font=\tiny\bfseries\itshape, align=center, inner sep=2pt},
    loopicon/.style={font=\tiny, inner sep=1pt, anchor=north east},
]
% Start
\node (startnode) [circle, fill=black, inner sep=1.5pt] {};

% Phases
\node (p1) [phase, right=of startnode, xshift=1.4cm] {Phase 1:\\Well-founded Operator \\ Def.  \ref{def:well-foundedoperator}};
\node [loopicon] at ($(p1.north east) + (-0.25,-0.15)$) {$\circlearrowright$};

\node (p2) [phase, right=of p1, xshift=0.8cm] {Phase 2:\\Unit Propagation  \\ Def. \ref{def:UP-operators}};
\node [loopicon] at ($(p2.north east) + (-0.25,-0.15)$) {$\circlearrowright$};

\node (p3) [phase, right=of p2, xshift=0.8cm] {Phase 3:\\Guess and Check\\ \scriptsize(Compute All Stable Partitions)};
\node [loopicon, font=\tiny\bfseries] at ($(p3.north east) + (-0.1,-0.1)$) {$\forall$};

% Decisions
\node (d1) [decision, below=of p1, yshift=.3cm] {Prop.~\ref{propos:computingstablepartitions}\\holds?};
\node (d2) [decision, below=of p2, yshift=0.3cm] {Prop.~\ref{propos:computingstablepartitions}\\holds?};

% New Intermediate Result & Prop 4.3 Check
\node (stable-res) [result, below=of p3, yshift=0.3cm] {Stable Partitions};
\node (d43) [decision, below=of stable-res, yshift=0.3cm] {Def. \ref{def:well-foundedpartition}\\holds?};

% Final result
\node (res) [result, below=of p2, yshift=-1.5cm] {Well-founded Partition};
\node (fail-msg) [font=\tiny\bfseries, color=red!70!black, right=of d43, xshift=0.5cm] {No WFP};

% Entry arrow
\draw [arrow] (startnode) -- node[above, label] {Input: \(\mathcal{K}_G\)} node[below, label] {Start with \((\emptyset,\emptyset)\)} (p1);

% Phase-to-decision arrows
\draw [arrow] (p1) -- node[right, label, xshift=-18mm] {Output:\((T_\omega,F_\omega)\)} (d1);
\draw [arrow] (p2) -- node[right, label, xshift=-18mm] {Output:\((T_E,F_E)\)} (d2);

% YES branches
\draw [arrow, dashed] (d1.south) -- node[left, label] {yes} (res.north west);
\draw [arrow, dashed] (d2.south) -- node[right, label] {yes} (res.north);

% NO branches to next phase (go forward)
\draw [arrow] (d1.east) -- ++(1.4,0) |- node[pos=0.30, yshift=-14mm,  xshift=-8.5mm,above, label] {no} node[pos=0.50, yshift=-10mm,  xshift=-8.5mm, below, label] {Input: \((T
_\omega,F_\omega)\)} (p2.west);
\draw [arrow] (d2.east) -- ++(1.4,0) |- node[pos=0.30, yshift=-14mm,  xshift=-8.5mm,above, label] {no} node[pos=0.50, yshift=-10mm,  xshift=-8.5mm, below, label] {Input: \((T
_E,F_E)\)} (p3.west);
% Prop 4.2 YES -> Stable Partitions -> Prop 4.3
\draw [arrow] (p3) -- node[right, label] {yes} (stable-res);
\draw [arrow] (stable-res) -- (d43);
\draw [arrow, dashed] (d43.west) -| node[above, label, pos=0.2] {yes} (res.east);

% NO branch for 4.3
\draw [arrow] (d43.east) -- node[above, label] {no} (fail-msg);
\end{tikzpicture}
\end{flushleft}
\vspace{-1em}
\caption{Workflow for well-founded partition computation of a \(\hybridmknf\) knowledge base.}
\end{figure}
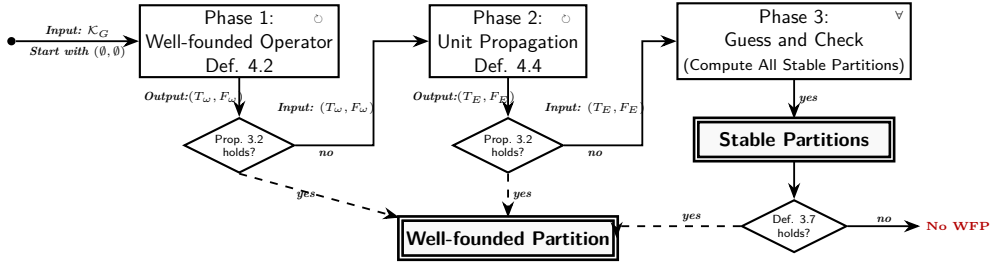

\subsection{Phase 1: Fixpoint Computation of Well-founded Operator}\label{sec:well-foundedoperator}
If the well-founded partition exists, the well-founded operator used here
computes it in most cases, with the remaining cases handled by subsequent
phases. This operator, originally introduced by \cite{ji2017wellfoundedoperatorsnormalhybrid},
is defined with respect to a partial partition \((T,F)\).
 It combines
(i) the operator $T^F_{\rKG}(T)$ from Def.~\ref{def:monotoncoepratos}, which computes the
modal atoms derivable from the $\hybridmknf$ knowledge base, and
(ii) the operator $F_{\rKG}$ (defined below), also called the unfounded set operator, which computes the
largest set of modal atoms that are not derivable from the $\hybridmknf$ knowledge base
with respect to $(T,F)$.

A modal atom $\mathbf{K}H \in \EKAKG$ belongs to an unfounded set if it does not admit a valid
derivation from the rule component $\mathcal{P}$ together with the DL ontology $\mathcal{O}$.
Formally, $\mathbf{K}H$ is not derivable with respect to $(T,F)$ if
\(
\OBO{\EKAKG \setminus F} \not\models H,
\)
and for every rule $r \in \rPG$ with head $\mathbf{K}H$, the body of $r$ evaluates to false
under $(T,F)$, that is, $b(r)[T,F] = \mathbf{f}$. However, when $\OBO{\EKAKG \setminus F}$ is unsatisfiable, the principle of explosion applies,
and every formula becomes derivable from the objective knowledge. Consequently, this characterization of non-derivability no longer applies, and the unfounded set operator $F_{\rKG}$ cannot be defined in a straightforward manner.

We therefore adopt the unfounded set construction for hybrid MKNF knowledge bases from \cite{ji2017wellfoundedoperatorsnormalhybrid} and adapt it to $\hybridmknf$. Let $(T,F)$ be a partial partition of \(\EKAKG\), and let $\mathbf{K}H$ be a ground modal atom.
For any set of rules $R \subseteq \rPG$,
$\mathsf{head}(R)$ is defined as $\{\, \xi \mid \mathbf{K}\xi = \mathsf{h}(r) \text{ for some } r \in R \,\}$.

\begin{definition}[Unfounded set]\label{def:unfounded set}
Let $\rKOPG$ be a ground \( \hybridmknf\) knowledge base, and let $(T,F)$ be a partial partition of $\EKAKG$.
A set $UF \subseteq \EKAKG$ is an \emph{unfounded set} with respect to \ $(T,F)$ if, for every $\mathbf{K}H \in UF$, where $H$ is either $A$ or $\neg A$, at least one of the following holds:
\begin{itemize}[leftmargin=*, itemsep=.3pt]
 \item $\OBO{T} \models \overline{H}$
  \item for every $R \subseteq \rPG$ such that
  \begin{enumerate} 
   \item $\mathrm{head}(R) \cup \OBO{T} \models H$, and for all $R' \subset R$, $\mathrm{head}(R') \cup \OBO{T} \not\models H$, and 
       \item for each $\mathbf{K}\xi \in F$, the set
        $\mathrm{head}(R) \cup \OBO{T} \cup \{\overline{\xi}\}$ is satisfiable;
  \end{enumerate}
  there exists a rule $r \in R$ of the form $h(r) \gets b^+(r), b^-(r)$ such that at least one of the following holds: $\mathbf{K}(b^-(r)) \cap T \neq \emptyset$,  $b^+(r) \cap F \neq \emptyset$, or $b^+(r) \cap UF \neq \emptyset$.
\end{itemize}
The \emph{greatest unfounded set} of $\rKG$ with respect to\ $(T,F)$, denoted $F_{\rKG}(T,F)$, is the largest unfounded set satisfying the above condition.
\end{definition}
Computing unfounded sets requires considering all possible derivations of a modal atom
obtained by combining the rule component with the DL ontology, and verifying that none of
these derivations succeeds with respect to the current partial partition (see
Ex.~\ref{ex:unfoundedsetexample} in the appendix). For a $\hybridmknf$ knowledge base,
this task is computationally expensive in the general case.
Since our objective is to define a semantics for $\hybridmknf$, we do not address the
computation of unfounded sets and rely on this notion solely for theoretical
purposes. However, we can either adopt the unfounded set computation proposed in~\cite{ji2017wellfoundedoperatorsnormalhybrid} or restrict the DL part to Datalog-rewritable DLs, in which case unfounded sets can be computed by evaluating the rewritten rules.

Based on Def.~\ref{def:unfounded set} of \(F_{\rKG}(T,F)\) and
Def.~\ref{def:monotoncoepratos} of \(T^F_{\rKG}(T)\), we define the
well-founded operator as follows:

\begin{definition}[Well-founded Operator]\label{def:well-foundedoperator}
Let \(\rKOPG\) be a ground \(\hybridmknf\) knowledge base, and let \((T,F)\) be a partial partition of \(\EKAKG\). The well-founded operator $W_{\rKG}$ with respect to $(T,F)$ is defined as 
\(
W_{\rKG}(T, F)
   =
   ( T^F_{\rKG}(T)
         ,\;
        F_{\rKG}{(T,F)}
).
\)
\end{definition}

With respect to non contradictory partial partitions (Def. \ref{def:noncontradictory}), $W_{\rKG}$ is monotonic and then has a least fixpoint. We next define when a partial partition is contradictory. This notion of contradiction can be used to detect MKNF inconsistency of a $\hybridmknf$ knowledge base.

\begin{definition}\label{def:noncontradictory}
Let $\rKOPG$ be a ground $\hybridmknf$ knowledge base and let $T,F \subseteq \EKAKG$.
$(T,F)$ is \emph{non-contradictory} if and only if
$(T,F)$ is a partial partition and $\OBO{T}$ is satisfiable.
Otherwise, $(T,F)$ is said to be \emph{contradictory}.
\end{definition}

%$T \cap F = \emptyset$ and 

Then the operator $W_{\rKG}$ is monotonic
\cite{ji2017wellfoundedoperatorsnormalhybrid} with respect to the following order defined on non-contradictory partial partitions: for
$(T_1,F_1)$ and $(T_2,F_2)$, we write
$(T_1,F_1)\subseteq(T_2,F_2)$ if $T_1\subseteq T_2$ and $F_1\subseteq F_2$. That is, whenever
$(T,F)\subseteq(T',F')$ and both partitions are non-contradictory, it holds that
$W_{\rKG}(T,F)\subseteq W_{\rKG}(T',F')$.
Alg.~\ref{alg:wfm-stable} computes the least fixpoint of $W_{\rKG}$ by
iterative construction starting from the initial partition
$(\emptyset,\emptyset)$.
Since $T$ and $F$ range over subsets of the finite set $\EKAKG$, convergence is
guaranteed after finitely many steps, yielding the fixpoint
$W_{\rKG}\uparrow\omega=(T_\omega,F_\omega)$. Alg~\ref{alg:wfm-stable} checks whether $(T_\omega, F_\omega)$ is a stable partition (Prop.~\ref{propos:computingstablepartitions}), terminating if the condition is satisfied and  otherwise proceeding to Phase~2.

\begin{algorithm}[!htp]\footnotesize
\DontPrintSemicolon
\caption{Fixpoint Computation of  $W_{\rKG}$: $\mathsf{Out}(\rKG)\in\{\bot,(T_W,F_W)\}$ }
\label{alg:wfm-stable}
$i \gets 0$, $(T_0,F_0) \gets (\emptyset,\emptyset)$\;
\Repeat{$(T_i,F_i) = (T_{i-1},F_{i-1})$}{
  $(T_{i+1},F_{i+1}) \gets W_{\rKG}(T_i,F_i)$  (Def. \ref{def:well-foundedoperator}) 
  
  \If{$(T_{i+1},F_{i+1})$ is contradictory  (Def. \ref{def:noncontradictory})  }{
    \Return $\bot$\tcp*[r]{MKNF-inconsistent (Thm.  \ref{thm:finalreasult})}  
  }
  $i \gets i+1$\;
}
$(T_\omega,F_\omega) \gets (T_i,F_i)$\;

\If{$(T_\omega,F_\omega)$ is a stable partition (Prop. \ref{propos:computingstablepartitions}) }{
  \Return $(T_W,F_W)\gets(T_\omega,F_\omega)$ \tcp*[r]{Otherwise, continue with Phase 2}
}

\end{algorithm}
The following proposition relates the partial partition $(T_\omega, F_\omega)$ computed as the fixpoint of $W_{\rKG}$ to all stable partitions of a $\hybridmknf$ knowledge base.

\begin{proposition}\label{prop:trueandfalseineverythreevaluedwellfoudnedoperaptr}
Let \(\rKOPG\) be a ground \(\hybridmknf\) knowledge base and
\((T_\omega,F_\omega)\) a fixpoint of the well-founded operator \(W_{\rKG}\).
For any modal atom \(\mathbf{K}H \in \EKAKG\), with \(H\) of the form \(\xi\) or \(\neg\xi\),
and for every stable partition \((T,F)\) of \(\rKG\), the following holds:
\(\mathbf{K}H \in T_\omega\) implies \(\mathbf{K}H \in T\), and
\(\mathbf{K}H \in F_\omega\) implies \(\mathbf{K}H \in F\).
\end{proposition}

\begin{proof}
Proof is provided in Prop. \ref{prop:trueandfalseineverythreevaluedwellfoudnedoperaptra1}.
\end{proof}

\subsection{\textbf{Phase 2: Unit Propagation. }}\label{sec:unitpropagation}
Unit propagation is applied when Phase~1 fails to compute a stable partition.
The following example illustrates a situation in which this phase becomes
necessary.

\begin{example}\label{ex:example4unit-prpopatiton1}
Consider the $\hybridmknf$ knowledge base $\mathcal{K}_1=(\rO_1,\rP_1)$, where
$\pi(\rO_1)=\{c\}$ and
\(
\rP_1=\{
r_1:\mathbf{K}a \leftarrow \mathbf{not}\ b.,\ 
r_2:\mathbf{K}b \leftarrow \mathbf{not}\ a.,\ 
r_3:\mathbf{K}\neg c \leftarrow \mathbf{K}a.
\}.
\)
The fixpoint of $W_{\rK_1}$ is
$(T_{\omega_1},F_{\omega_1})=(\emptyset,\{\mathbf{K}\neg c\})$. By Proposition~\ref{propos:computingstablepartitions},
$(T_{\omega_1},F_{\omega_1})$ is not stable, since
$\OBO{\Gamma(\EKAKG \setminus T_{\omega_1})}$ is unsatisfiable
(with $\Gamma(\EKAKG \setminus T_{\omega_1})=\{\mathbf{K}a,\mathbf{K}b,\mathbf{K}\neg c\}$). 
\emph{Reason.} This partition is not a stable partition because of rule $r_3$.
With respect to $(T_{\omega_1},F_{\omega_1})$, the head
$h(r_3)=\mathbf{K}\neg c$ is false, while its body
$b(r_3)=\mathbf{K}a$ is undefined (Lemma~\ref{lem:3valued-eval});
therefore, the rule is not satisfied
(Lemma~\ref{lemma:proofstbaleparitiomn2}).  
As the head of $r_3$ is false, satisfaction requires its body to be false as well. Applying unit propagation \cite{ji2017wellfoundedoperatorsnormalhybrid} enforces $\mathbf{K}a$ to be false, which in turn allows rule $r_2$ to derive $\mathbf{K}b$ as true.  Thus, from $(\emptyset,\{\mathbf{K}\neg c\})$ we obtain the partial partition
$(\{\mathbf{K}b\},\{\mathbf{K}a,\mathbf{K}\neg c\})$ which is stable according to Prop. \ref{propos:computingstablepartitions}.
\end{example}
More precisely, let $r \in \mathcal{P}_G$ and let $(T,F)$ be a partial partition of
$\EKAKG$ such that $h(r)[(T,F)]$ evaluates to \emph{false} and
$b(r)[(T,F)]$ evaluates to \emph{undefined}.
If there exists exactly one modal atom $\mathbf{K}L \in b^+(r) \cup \mathbf{K}(b^-(r))$
such that $\mathbf{K}L \notin T \cup F$, then the truth value of
$\mathbf{K}L$ is forced in order to satisfy the rule.
In particular, $\mathbf{K}L$ must be assigned \emph{false} if
$\mathbf{K}L \in b^+(r)$, and \emph{true} if
$\mathbf{K}L \in \mathbf{K}(b^-(r))$.
This intuition is formalized in Definition~\ref{def:UP-operators}
by introducing two unit propagation operators.

\begin{definition}[Unit Propagation Operators]\label{def:UP-operators}
Let $\rKOPG$ be a ground $\hybridmknf$ knowledge base, and let
$(T_\omega, F_\omega)$ denote the fixpoint of the operator $W_{\rKG}$.
Let $\mathbf{K}a \in \EKAKG$, where $a$ is either an atom $A$ or its classical
negation $\neg A$.
For any sets \(X,Y \subseteq \EKAKG\), we define two unit-propagation operators:
\begin{align*}\small
UPT^{(T_\omega,F_\omega)}_{\mathcal K_G}(X,Y)
= {} &
\{\, \mathbf K a \mid \exists r \in \mathcal P_G \text{ such that }\;
\OBO{ T_\omega \cup X} \models \overline{h(r)}, \ \ b^+(r) \subseteq T_\omega \cup X,\;
\nonumber\\
&\qquad
\mathbf K(b^-(r)) \cap (T_\omega \cup X) = \emptyset\text{ and } \ \ \mathbf K(b^-(r)) \setminus (F_\omega \cup Y) = \{\mathbf K a\}
\,\},
\\[.4ex]\small
UPF^{(T_\omega,F_\omega)}_{\mathcal K_G}(X,Y)
= {} &
\{\, \mathbf K a \mid \exists r \in \mathcal P_G \text{ such that }\;
\OBO{ T_\omega \cup X} \models \overline{h(r)}, \;
\mathbf K(b^-(r)) \subseteq F_\omega \cup Y,
\nonumber\\
&\qquad
 b^+(r) \cap (F_\omega \cup Y) = \emptyset\text{ and } b^+(r) \setminus (T_\omega \cup X) = \{\mathbf K a\}\;
\,\}.
\end{align*}
\end{definition}
The atoms newly derived by unit propagation may, in turn, enable additional consequences stemming from both the DL ontology and the rule component.
To capture this interaction, we introduce the extending operator $E_{\rKG}$.

\begin{definition}\label{def:E-operator}
Let $\rKOPG$ be a ground $\hybridmknf$ knowledge base, and let
$(T_\omega,F_\omega)$ be the fixpoint of $W_{\rKG}$.
For $X,Y \subseteq \EKAKG$, we define
\small
\[
E^{(T_\omega,F_\omega)}_{\rKG}(X,Y)
=
\bigl(
UPT^{(T_\omega,F_\omega)}_{\rKG}(X,Y)\ \cup\ T^Y_{\rKG}(X),\ 
UPF^{(T_\omega,F_\omega)}_{\rKG}(X,Y)\ \cup\ F_{\rKG}(X,Y)
\bigr).
\]
\normalsize
\end{definition}

The operator $E^{(T_\omega,F_\omega)}_{\rKG}$ is monotonic with respect to set
inclusion, analogously to the well-founded operator.
Hence, iterative application from the empty partition
$(\emptyset,\emptyset)$ is guaranteed to converge to a least fixpoint
$E^{(T_\omega,F_\omega)}_{\rKG}\uparrow\omega=(T_E,F_E)$ after finitely many
steps.
Alg.~\ref{alg:wfm-stable2} iteratively computes this fixpoint and returns
$(T_E,F_E)$ if it is stable, otherwise proceeding to Phase~3.
\begin{algorithm}[!htp]\footnotesize
\DontPrintSemicolon
\setcounter{AlgoLine}{10}
\caption{Fixpoint Computation of \(E_{\rKG}\): $\mathsf{Out}(\rKG)\in\{\bot,(T_W,F_W)\}$ }
\label{alg:wfm-stable2}
$j \gets 0$, $(T^\prime_0,F^\prime_0) \gets (\emptyset,\emptyset)$\;
\Repeat{$(T^\prime_j,F^\prime_j) = (T^\prime_{j-1},F^\prime_{j-1})$}{
  $(T^\prime_{j+1},F^\prime_{j+1}) \gets E^{(T_\omega,F_\omega)}_{\rKG}(T^\prime_j,F^\prime_j)$(Def. \ref{def:UP-operators})
  
  \If{ $(T^\prime_{j+1},F^\prime_{j+1})$ is contradictory (Def. \ref{def:noncontradictory})}{
    \Return $\bot$\tcp*[r]{MKNF-inconsistent (Thm.  \ref{thm:finalreasult})}
  }
  $j \gets j+1$\;
}
$(T_E,F_E) \gets (T^\prime_j,F^\prime_j)$\;
\If{$(T_E,F_E)$ is a stable partition (Prop. \ref{propos:computingstablepartitions})}{
  \Return $(T_W,F_W)\gets(T_E,F_E)$ \tcp*[r]{Otherwise, continue with Phase 3}
}
\end{algorithm}

The Prop. \ref{prop:trueandfalseineverythreevaluedunit} states that the partial partition $(T_E, F_E)$,
computed as the fixpoint of $E_{\rKG}$, is contained in all stable partitions
of a $\hybridmknf$ knowledge base.
\begin{proposition}\label{prop:trueandfalseineverythreevaluedunit}
Let \(\rKOPG\) be a ground \(\hybridmknf\) knowledge base and
\((T_E,F_E)\) a fixpoint of \(E_{\rKG}\).
For any modal atom \(\mathbf{K}H \in \EKAKG\), with \(H\) of the form \(\xi\) or \(\neg\xi\),
and for every stable partition \((T,F)\) of \(\rKG\), the following holds:
\(\mathbf{K}H \in T_E\) implies \(\mathbf{K}H \in T\), and
\(\mathbf{K}H \in F_E\) implies \(\mathbf{K}H \in F\).
\end{proposition}
\begin{proof}
Proof is provided in Prop. \ref{prop:trueandfalseineverythreevaluedunita2}.
\end{proof}
\vspace{-0.8em}
\subsection{\textbf{Phase 3: Guess-and-Check Approach. }}\label{sec:guessandcheck}

In some $\hybridmknf$ knowledge bases, Phases~1 and~2 are insufficient to compute
the well-founded partition. In such cases, we employ a guess-and-check phase to
ensure the completeness of the overall procedure.
We illustrate this with a simple example.

\begin{example}\label{ex:guessandcheck}
Consider the $\hybridmknf$ knowledge base $\mathcal{K}_2 = (\rO_2,\rP_2)$, where
$\pi(\rO_2) = \{c\}$ and
\(
\rP_2 = \{
r_1:\mathbf{K}a \gets \mathbf{not}\ b.,\ 
r_2:\mathbf{K}b \gets \mathbf{not}\ a.,\ 
r_3:\mathbf{K}\neg c \gets \mathbf{K}d,\mathbf{K}a.,\ 
r_4:\mathbf{K}d \gets \mathbf{K}a,\mathbf{not}\ d.
\}.
\)
$(T_{\omega_2},F_{\omega_2}) = (\emptyset,\{\mathbf{K}\neg c\})$ and $(T_{E_2},F_{E_2}) = (\emptyset,\{\mathbf{K}\neg c\})$. The objective knowledge
$\OBO{\Gamma(\EKAKG \setminus T_{E_2})}$
is unsatisfiable (with \(\Gamma(\EKAKG \setminus T_{E_2})=\{\mathbf{K}a,\mathbf{K}b,\mathbf{K}\neg c,\mathbf{K}d\}\)).
Hence, neither $(T_{\omega_2},F_{\omega_2})$ nor
$(T_{E_2},F_{E_2})$ is a stable partition. To compute the well-founded partition, we must therefore
enumerate all stable partitions and select the minimal one
according to Def.~\ref{def:well-foundedpartition}.  $\mathcal{K}_2$ admits two stable partitions
extending $(T_{E_2},F_{E_2})$ and satisfying
Prop.~\ref{propos:computingstablepartitions}: \((T_1,F_1) = (\{\mathbf{K}b\},\{\mathbf{K}\neg c, \mathbf{K}a, \mathbf{K}d\})\) and  \((T_2,F_2) = (\{\mathbf{K}b\},\{\mathbf{K}\neg c, \mathbf{K}a\})\). 
The well-founded partition of $\mathcal{K}_2$ is
$(T_2,F_2)$, according to
Def.~\ref{def:well-foundedpartition} (\(T_2  = T_1\) and \(F_2 \subseteq F_1\)).
\end{example}

Prop. \ref{prop:trueandfalseineverythreevaluedwellfoudnedoperaptr} and Prop.~\ref{prop:trueandfalseineverythreevaluedunit} guarantees that the partial
partition $(T_{E}, F_{E})$, although not stable, is contained in every stable partition. 
Starting from this base partition (Ex.~\ref{ex:guessandcheck}), we therefore adopt a
\emph{guess-and-check} approach that enumerates all partial partitions extending
$(T_{E}, F_{E})$ and identifies those that are stable according to
Prop.\ref{propos:computingstablepartitions}.
This procedure is described in Alg.~\ref{alg:wfm-stable1}. 
\vspace{-0.8em}
\begin{algorithm}[!htp]\footnotesize
\DontPrintSemicolon
\setcounter{AlgoLine}{21}
\caption{Stable Partition Computation: $\mathsf{Out}(\mathcal K)\in\{\bot,(T_W,F_W),\texttt{NoWFM}\}$}
\label{alg:wfm-stable1}
$U \gets \EKAKG \setminus (T_E \cup F_E)$,\ \ $\mathcal{S}\gets\emptyset$\;

\ForEach{$T_u \subseteq U$}{
  \ForEach{$F_u \subseteq U \setminus T_u$}{
    $T_S \gets T_E \cup T_u$,\ \ $F_S \gets F_E \cup F_u$\;
    \If{$(T_S,F_S)$ stable (Prop.~\ref{propos:computingstablepartitions})}{
      $\mathcal{S} \gets \mathcal{S} \cup \{(T_S,F_S)\}$\;
    }
  }
}

\If{$\mathcal{S}=\emptyset$}{\Return $\bot$ \tcp*[r]{MKNF-inconsistent}}

\If{$\exists\,(T,F)\in\mathcal{S}$ such that Def~\ref{def:well-foundedpartition} holds}{
  
  \Return $(T_W,F_W) \gets (T,F)$\;
}
\Return \texttt{"NoWFM"} \tcp*[r]{no well-founded model}
\end{algorithm}
\vspace{-1.6em}
\subsection{Well-founded Model}
Algorithms~\ref{alg:wfm-stable}, \ref{alg:wfm-stable2}, and~\ref{alg:wfm-stable1}
compute an output $\mathsf{Out}(\mathcal K)\in\{\bot,(T_W,F_W),\textnormal{\texttt{NoWFM}}\}$
for a ground $\hybridmknf$ knowledge base $\mathcal K$, where \((T_W,F_W)\) is stable partition, $\bot$ denotes
MKNF inconsistency and \texttt{NoWFM} denotes the absence of a unique stable
partition.

%\begin{proposition}[MKNF-Consistency]
%\label{prop:mknfconsistency}
%Let $\rKOPG$ be a ground $\hybridmknf$ knowledge base, and let
%$\mathsf{Out}(\rKG)$ denote the output of
%Algorithms~\ref{alg:wfm-stable}, \ref{alg:wfm-stable2}, and~\ref{alg:wfm-stable1}.
%If $\mathsf{Out}(\rKG)=\bot$, then $\rKG$ is MKNF-inconsistent.
%\end{proposition}

%\begin{theorem}\label{the:stableoutputheorem}
%Let $\rKOPG$ be a ground $\hybridmknf$ knowledge base, and let
%$\mathsf{Out}(\rKG)$ denote the output of
%Algorithms~\ref{alg:wfm-stable}, \ref{alg:wfm-stable2}, and~\ref{alg:wfm-stable1}.
%If $\mathsf{Out}(\rKG)=(T_W,F_W)$, then $(T_W,F_W)$ is a stable partition of
%$\rKG$.
%\end{theorem}

\begin{theorem}\label{thm:finalreasult}
Let $\rKOPG$ be a ground $\hybridmknf$ knowledge base, and let
$\mathsf{Out}(\rKG)$ denote the output of
Algorithms~\ref{alg:wfm-stable}, \ref{alg:wfm-stable2}, and~\ref{alg:wfm-stable1}.
(1) If $\mathsf{Out}(\rKG)=\bot$, then $\rKG$ is MKNF-inconsistent. (2) If $\mathsf{Out}(\rKG)=(T_W,F_W)$, then $(T_W,F_W)$ is the well-founded partition
of $\rKG$, and the induced three-valued MKNF interpretation $(M_W,N_W)$ is the
well-founded MKNF model of $\rKG$. (3) If $\mathsf{Out}(\rKG)=\textnormal{\texttt{NoWFM}}$, then $\rKG$ has no well-founded model.
\end{theorem}
\begin{proof}
Proof is provided in Thm. \ref{prop:well-foundedmodela1}.
\end{proof}

The data complexity of computing the well-founded partition of a
\(\hybridmknf\) knowledge base depends on the phase in which the stable
partition is obtained.
If it is computed in Phases~1 or~2, the complexity depends on the data
complexity of DL reasoning and unfounded set computation.
If it is only obtained in Phase~3, all extensions of the partial partition
\((T,F)\) obtained after Phase~2 over the remaining atoms
\(U =  \EKAKG  \setminus (T \cup F)\) are enumerated,
yielding at most \(3^{|U|}\) candidates; each candidate can be checked to
be a stable partition in polynomial time relative to DL reasoning \cite{LIU2017123},
resulting overall in exponential.
\begin{proposition}\label{props:complexity}
Let $\rKOPG$ be a DL-safe $\hybridmknf$ knowledge base. Assume data complexity $\mathcal{C}_1$ for computing unfounded sets with respect to  a partial partition $(T,F)$, and $\mathcal{C}_2$ for DL satisfiability and ground entailment.  Let $(T_W^1,F_W^1)$ and $(T_W^2,F_W^2)$ be the partitions obtained after Phase~2 and Phase~3, respectively. $(i)$ If $(T_W^1,F_W^1)$ coincides with the well-founded partition of $\rKG$, then it can be computed with data complexity $\mathbf{PTime}^{\mathcal{C}_1 \cup \mathcal{C}_2}$.
 $(ii)$ If $(T_W^2,F_W^2)$ coincides with the well-founded partition of $\rKG$ then  it can be computed with data complexity is $\mathbf{EXPTime}^{\mathbf{P}^{\mathcal{C}_1}}$.

\end{proposition}

\section{Discussion and Related Works}
Table~\ref{tab:comparison} summarizes the comparison between $\hybridmknf$,
\cite{Knorr2011LocalCW}, and \cite{LIU2017123}. The approach of \cite{Knorr2011LocalCW} applies only to coherent hybrid MKNF
knowledge bases, where the well-founded partition is obtained by an alternating
fixpoint construction (based on \cite{Gelder1989TheAF}). Another related work is~\cite{ji2017wellfoundedoperatorsnormalhybrid}, which
introduces a well-founded operator that we adapt in this work.
However, its primary goal is constraint propagation for the efficient computation
of two-valued MKNF models.
Moreover, it is shown in~\cite{ji2017wellfoundedoperatorsnormalhybrid} that the
resulting well-founded partition coincides with one obtained
by~\cite{Knorr2011LocalCW}.
\vspace{-0.7em}
 \begin{table}[htp] \scriptsize
\centering
\label{tab:comparison}
\begin{tblr}{
  column{even} = {c},
  column{3} = {c},
  hlines,
  vlines,
}
\textbf{Aspect} & $\hybridmknf$                                                                                      & \cite{Knorr2011LocalCW}         & \cite{LIU2017123}             \\
Knowledge bases & All Hybrid MKNF$\neg$                                                                                              & {Coherent Hybrid \\MKNF}       & All Hybrid MKNF                \\
Rule component  & with classical negation                                                                                      & {without classical\\~negation} & {without classical\\~negation} \\
{Well-founded \\ partition \\ Computation}     & {Fixpoint computation \\with unit propagation\\fallback to guess-and-check}                                  & {Fixpoint \\computation}       & Guess-and-check                \\
Complexity      & Prop.\ref{props:complexity} & $\mathsf{PTime}^{\mathcal{C}_2}$                        & { Existence of a stable \\ partition:   $\mathsf{NPTime}^{\mathsf{PTime}^{\mathcal{C}_2}} $ }          
\end{tblr}
\caption{Comparison with Related Approaches}
\vspace{-0.9em}
\end{table}

The main advantage of our approach is that it supports classical negation in the rule component while providing a general procedure for computing the well-founded model of arbitrary $\hybridmknf$ knowledge bases. A current limitation is that we do not present a method for unfounded set computation, although existing methods for hybrid MKNF can be adopted. Additionally, the proposed computation methodology may require a guess-and-check approach (phase 3), which is exponential in the worst case. This phase is nevertheless unavoidable to ensure correctness and completeness, as Phases~1 and~2  are insufficient to compute the well-founded partition for some $\hybridmknf$ knowledge bases.

A $\hybridmknf$ reasoner can be implemented using either bottom-up or top-down
approaches. Assuming unfounded set computation method exists, a bottom-up solver applies the fixpoint operators of
Algorithms~\ref{alg:wfm-stable}--\ref{alg:wfm-stable1} directly.  But this approach may be inefficient in applications with frequent knowledge base updates.  In the top-down approach, one option is to translate the $\hybridmknf$ knowledge base into a logic program with oracle predicates for DL calls, together with transformations for enforcing coherency principle and  detect contradictions.
A formal proof is then required to establish that the well-founded model of the resulting program corresponds to the well-founded partition of the $\hybridmknf$ knowledge base. This would enable the use of SLG resolution extended with oracle calls, similar to \cite{Alferes2010QueryDrivenPF}. 
However, this approach supports only Phases~1 and~2 of the computation;
Phase~3 requires global reasoning and cannot be supported. Thus, the approach is complete only for knowledge bases whose well-founded partition is obtained in Phase~1 or Phase~2.
Another possible strategy is to define an abstract solver in which the fixpoint operators are evaluated over a reified representation of the rule component, while interacting with the DL reasoner on demand. This allows the fixpoint to be computed through a recursive, top-down evaluation, similar to~\cite{Gomes2010ImplementingQA}.

\section{Conclusion and Future Work}

In this work, we introduced $\hybridmknf$, an extension of hybrid MKNF knowledge
bases that supports classical negation in the rule component, together with a
general methodology for computing its well-founded model. The proposed approach
applies to arbitrary $\hybridmknf$ knowledge bases and is based on the
computation of a well-founded partition through a three-phase process. 
Building on this foundation, we plan to develop a query answering system for $\hybridmknf$ knowledge bases.

\bibliographystyle{eptcs}
\bibliography{generic}
\newpage
\appendix

\section{Examples} \label{sec:illustrativeexamples}

\begin{example}[Phase 1 Example] \label{ex:example1}
Consider a $\hybridmknf$ knowledge base
$\mathcal{K}= (\rO,\rP)$, where 
$\pi(\rO)= \pred{f} \wedge (\pred{b} \supset \neg \pred{d})$ and
\[
\begin{aligned}
\rP_1 = \{\, 
r_1:\mathbf{K}\neg \pred{a} \leftarrow .
r_2:\mathbf{K}\pred{a} \leftarrow \mathbf{not}\ \pred{b}. \ \ 
r_3:\mathbf{K}\pred{b} \leftarrow \mathbf{not}\ \pred{a}. \ \ 
r_4:\mathbf{K}\pred{c} \leftarrow \mathbf{K}\pred{a}.\\
r_5:\mathbf{K}\pred{e} \leftarrow \mathbf{K}\neg \pred{d},\, \mathbf{K}\pred{f}.  \ \ 
r_6:\mathbf{K}\pred{g} \leftarrow \mathbf{not}\ \pred{g}.
\,\}.
\end{aligned}
\]

Consider the three-valued MKNF interpretation of $\pi(\mathcal{K})$ given by
{\small
\[
(M_1,N_1) =
\big(
\{
\{\pred{f},\pred{e},\pred{b}\},
\{\pred{f},\pred{e},\pred{b},\pred{g}\},
\{\pred{f},\pred{e},\pred{b},\pred{c}\},
\{\pred{f},\pred{e},\pred{b},\pred{g},\pred{c}\}
\},
\{
\{\pred{f},\pred{e},\pred{b},\pred{g}\},
\{\pred{f},\pred{e},\pred{b},\pred{g},\pred{c}\}
\}
\big).
\]
}
\normalsize
It holds that $(M_1,N_1) \models \pi(\mathcal{K}_1)$, and there exists no
$(M_1',N_1')$ such that $M_1 \subset M_1'$ and $N_1 \subset N_1'$ with
$(M_1',N_1') \models \pi(\mathcal{K}_1)$.
Therefore, $(M_1,N_1)$ is a three-valued MKNF model of $\mathcal{K}_1$.
With respect to $(M_1,N_1)$, \(\mathbf{K}\neg a, \mathbf{not}\ a, \mathbf{K}f,\mathbf{K}e,\mathbf{K}\neg d,\mathbf{K}b \text{ is }\mathbf{t}\), \(,\mathbf{K}\ c \text{ is } \mathbf{f}\) 
and \(\mathbf{K}g\text{ is }\mathbf{u}.\)
\end{example}

\begin{example}[Unfounded Set, Example~\ref{ex:example1} continued] \label{ex:unfoundedsetexample}
Let $\mathcal{K}=(\rO,\rP)$ and \\ $(T,F)=(\{\mathbf{K}\neg a,\mathbf{K}f\},\emptyset)$. 
Since $\OBO{\{\mathbf{K}\neg a,\mathbf{K}f \}}\models\neg a$, $\mathbf{K}a$ is unfounded with respect to\ $(T,F)$.
Moreover, the only rule supporting $\mathbf{K}c$ is $R=\{\mathbf{K}c\leftarrow\mathbf{K}a\}$, for which
$\OBO{T}\cup\mathrm{head}(R)\models c$, while $r\in R$ we have $\mathbf{K}a\in b^{+}(r)$ with $\mathbf{K}a$ is unfounded.
Hence, $\mathbf{K}c$ is also unfounded. So
\(
F^{(\{\mathbf{K}\neg a,\mathbf{K}f\},\emptyset)}_{\rKG}=\{\mathbf{K}a,\mathbf{K}c\}.
\)
\end{example}

\begin{example}[Example \ref{ex:example1} continued] \label{ex:well-foundedoperator}
    Consider a $\hybridmknf$ knowledge base
$\mathcal{K} = (\rO,\rP)$. 

\(\EKAKG =  \{\mathbf{K}\neg \pred{a},\mathbf{K}\pred{a},\mathbf{K}\pred{b},\mathbf{K}\pred{d}, \mathbf{K}\pred{f},\mathbf{K}\pred{e},\mathbf{K}\neg \pred{d}, \mathbf{K}\pred{g} \}\).
Computing the well-founded partition of \(\rK\) as follows: 
\[
\begin{aligned}
W_{\rK}\!\uparrow\!0 &= (\emptyset,\emptyset),\\
W_{\rK}\!\uparrow\!1 &= (\{\mathbf{K}\neg a, \mathbf{K}f\},\emptyset),\\
W_{\rK}\!\uparrow\!2 &= (\{\mathbf{K}\neg a, \mathbf{K}f\},\{\mathbf{K}a, \mathbf{K}c\}),\\
W_{\rK}\!\uparrow\!3 &= (\{\mathbf{K}\neg a, \mathbf{K}f, \mathbf{K}b\},\{\mathbf{K}a, \mathbf{K}c\}),\\
W_{\rK}\!\uparrow\!4 &= (\{\mathbf{K}\neg a, \mathbf{K}f, \mathbf{K}b, \mathbf{K}\neg d\},\{\mathbf{K}a, \mathbf{K}c\}),\\
W_{\rK}\!\uparrow\!5 &= (\{\mathbf{K}\neg a, \mathbf{K}f, \mathbf{K}b, \mathbf{K}\neg d, \mathbf{K}e\},\{\mathbf{K}a, \mathbf{K}c, \mathbf{K}d\}),\\
W_{\rK}\!\uparrow\!6 &= (\{\mathbf{K}\neg a, \mathbf{K}f, \mathbf{K}b, \mathbf{K}\neg d, \mathbf{K}e\},\{\mathbf{K}a, \mathbf{K}c, \mathbf{K}d\}).
\end{aligned}
\]
Since $W_{\rK}\!\uparrow\!5 = W_{\rK}\!\uparrow\!6$, a fixpoint is reached.  

\((T_\omega,F_\omega) =  (\{\mathbf{K}\neg a, \mathbf{K}f, \mathbf{K}b, \mathbf{K}\neg d\},\{\mathbf{K}a, \mathbf{K}c, \mathbf{K}d\}) \). The objective knowledge
$\OBO{\Gamma(\EKAKG \setminus T_{\omega})}$
is satisfiable  (with \(\Gamma(\EKAKG \setminus T_{\omega})=\{\mathbf{K}\neg a, \mathbf{K}a, \mathbf{K}b,\mathbf{K} c,\mathbf{K}g, \mathbf{K}f, \mathbf{K}e,  \mathbf{K}\neg d\}\)).  So \((T_\omega,F_\omega)\) is stable partition, and also the well-founded partition.  The well-founded model \((M_W,N_W)\) where \(M_W = \{I \mid I \models \OBO{T_\omega}\}\) and \(N_W= \{I \mid I \models \OBO{\EKAKG \setminus F_\omega}\}\)

The well-founded model is 
{\small
\[
(M_W,N_W) =
\big(
\{
\{\pred{f},\pred{e},\pred{b}\},
\{\pred{f},\pred{e},\pred{b},\pred{g}\},
\{\pred{f},\pred{e},\pred{b},\pred{c}\},
\{\pred{f},\pred{e},\pred{b},\pred{g},\pred{c}\}
\},
\{
\{\pred{f},\pred{e},\pred{b},\pred{g}\},
\{\pred{f},\pred{e},\pred{b},\pred{g},\pred{c}\}
\}
\big).
\]
}
\end{example}

\begin{example}\label{ex:exampleunit}
  Consider the hybrid MKNF knowledge base
\(
\mathcal{K} = (\mathcal{O}, \mathcal{P}),
\)
where
\(
\pi(\mathcal{O}) = \neg a \wedge b
\)
and
\[
\mathcal{P} = \{
\mathbf{K}a \leftarrow \mathbf{K}b
\}.
\]
\(\EKAKG =  \{\mathbf{K}\pred{a},\mathbf{K}\pred{b}\}\).
We obtain

\[
\begin{aligned}
(T_0,F_0) = W_{\rK}\!\uparrow\!0 &= (\emptyset,\emptyset),\\
(T_1,F_1) = W_{\rK}\!\uparrow\!1 &= (\{\mathbf{K}\neg a, \mathbf{K}b\},\emptyset),\\
(T_2,F_2) = W_{\rK}\!\uparrow\!2 &= (\{\mathbf{K}a,\mathbf{K}\neg a, \mathbf{K}b\}, \{\mathbf{K}a\})
\end{aligned}
\]

Since $\OBO{T_2}$ is unsatisfiable and
$T_2\cap F_2\neq\emptyset$, so $\mathcal{K}$ is MKNF-inconsistent.
\end{example}

\begin{example}[Phase 2 Example]
    Consider the $\hybridmknf$ knowledge base $\mathcal{K}=(\mathcal{O},\mathcal{P})$, where
\(
\pi(\mathcal{O})=\{c\}
\)
and
\[
\mathcal{P}=\{\, 
r_1:\mathbf{K}a \leftarrow \mathbf{not}\ b,\ 
r_2:\mathbf{K}b \leftarrow \mathbf{not}\ a,\ 
r_3:\mathbf{K}\neg c \leftarrow \mathbf{K}a
\,\}.
\]

\(\EKAKG =  \{\mathbf{K}\pred{a},\mathbf{K}\pred{b}, \mathbf{K}\neg \pred{c} \}\).

\[
\begin{aligned}
W_{\rK}\!\uparrow\!0 &= (\emptyset,\emptyset),\\
W_{\rK}\!\uparrow\!1 &= ( \emptyset, \{\mathbf{K}\neg c\}) \\
W_{\rK}\!\uparrow\!2 &= ( \emptyset, \{\mathbf{K}\neg c\})
\end{aligned}
\]

Since $W_{\rK}\!\uparrow\!2 = W_{\rK}\!\uparrow\!1$, a fixpoint is reached.   \((T_\omega,F_\omega) =  (\emptyset , \{\mathbf{K}\neg c\}) \).
$(T_{\omega},F_{\omega})$ is not stable, since
$\OBO{\Gamma(\EKAKG \setminus T_{\omega})}$ is unsatisfiable
(with $\Gamma(\EKAKG \setminus T_{\omega})=\{\mathbf{K}a,\mathbf{K}b,\mathbf{K}\neg c\}$). 
\[
\begin{aligned}
E^{(T_{\omega},F_{\omega})}_{\rK}\!\uparrow\!0 &= (\emptyset,\emptyset),\\
E^{(T_{\omega},F_{\omega})}_{\rK}\!\uparrow\!1 &= ( \emptyset, \{\mathbf{K}a,\mathbf{K}\neg c\}) \\
E^{(T_{\omega},F_{\omega})}_{\rK}\!\uparrow\!2 &= ( \{\mathbf{K}b \}, \{\mathbf{K}a,\mathbf{K}\neg c\}) \\
E^{(T_{\omega},F_{\omega})}_{\rK}\!\uparrow\!3 &= ( \{\mathbf{K}b \}, \{\mathbf{K}a,\mathbf{K}\neg c\})
\end{aligned}
\]
Since $E^{(T_{\omega},F_{\omega})}_{\rK}\!\uparrow\!3 = E^{(T_{\omega},F_{\omega})}_{\rK}\!\uparrow\!2$, a fixpoint is reached.   \((T_E,F_E) =  ( \{\mathbf{K}b \}, \{\mathbf{K}a,\mathbf{K}\neg c\})\) which is stable partition and also the well-founded partition.  The well-founded model \((M_W,N_W)\) where \(M_W = \{I \mid I \models \OBO{T_E}\}\) and \(N_W = \{I \mid I \models \OBO{\EKAKG \setminus F_E}\}\). 
The well-founded model is 
{\small
\[
(M_W,N_W) =
\big(
\{
\{\pred{c},\pred{b}\},
\{\pred{c},\pred{b},\pred{a}\}
\},
\{
\{\pred{c},\pred{b}\},
\{\pred{c},\pred{b},\pred{a}\}
\}
\big).
\]
}
\end{example}
\begin{example}[Inconsistent MKNF Example]
   Consider a $\hybridmknf$ knowledge base $\mathcal{K} = (\rO,\rP)$,where \(\pi(\rO) =  \{c\}\) and \[\rP = \{r_1:\mathbf{K}a \gets \mathbf{not \ } a., r_1:\mathbf{K}\neg c \gets \mathbf{not \ } a. \}\]
\(\EKAKG =  \{\mathbf{K}\pred{a}, \mathbf{K}\neg \pred{c} \}\).
   \[
\begin{aligned}
W_{\rK}\!\uparrow\!0 &= (\emptyset,\emptyset),\\
W_{\rK}\!\uparrow\!1 &= ( \emptyset, \{\mathbf{K}\neg c\}) \\
W_{\rK}\!\uparrow\!2 &= ( \emptyset, \{\mathbf{K}\neg c\})
\end{aligned}
\]

Since $W_{\rK}\!\uparrow\!2 = W_{\rK}\!\uparrow\!1$, a fixpoint is reached.   \((T_\omega,F_\omega) =  (\emptyset , \{\mathbf{K}\neg c\}) \).
$(T_{\omega},F_{\omega})$ is not stable, since
$\OBO{\Gamma(\EKAKG \setminus T_{\omega})}$ is unsatisfiable
(with $\Gamma(\EKAKG \setminus T_{\omega})=\{\mathbf{K}a,\mathbf{K}\neg c\}$). 
\[
\begin{aligned}
(T_0,F_0) =E^{(T_{\omega},F_{\omega})}_{\rK}\!\uparrow\!0 &= (\emptyset,\emptyset),\\
(T_1,F_1) =E^{(T_{\omega},F_{\omega})}_{\rK}\!\uparrow\!1 &= ( \{\mathbf{K}a\}, \{\mathbf{K}\neg c\}) \\
(T_2,F_2) =E^{(T_{\omega},F_{\omega})}_{\rK}\!\uparrow\!2 &= ( \{\mathbf{K}a\}, \{\mathbf{K}\neg c, \{\mathbf{K}a\}\}) 
\end{aligned}
\]
Since $\OBO{T_2}$ is satisfiable but 
$T_2\cap F_2\neq\emptyset$, so $\mathcal{K}_2$ is MKNF-inconsistent.
\end{example}

\begin{example}
Consider the $\hybridmknf$ knowledge base $\mathcal{K} = (\rO,\rP)$, where
$\pi(\rO) = \{c\}$ and
\[
\rP= \{
r_1:\mathbf{K}a \gets \mathbf{not}\ b,\ 
r_2:\mathbf{K}b \gets \mathbf{not}\ a,\ 
r_3:\mathbf{K}\neg c \gets \mathbf{K}d,\mathbf{K}a,\ 
r_4:\mathbf{K}d \gets \mathbf{K}a,\mathbf{not}\ d
\}.
\]
\(\EKAKG =  \{\mathbf{K}\pred{a},\mathbf{K}\pred{b}, \mathbf{K}\neg \pred{c},\mathbf{K}\pred{d} \}\).
\[
\begin{aligned}
W_{\rK}\!\uparrow\!0 &= (\emptyset,\emptyset),\\
W_{\rK}\!\uparrow\!1 &= ( \emptyset, \{\mathbf{K}\neg c\}) \\
W_{\rK}\!\uparrow\!2 &= ( \emptyset, \{\mathbf{K}\neg c\})
\end{aligned}
\]
Since $W_{\rK}\!\uparrow\!2 = W_{\rK}\!\uparrow\!1$, a fixpoint is reached.   \((T_\omega,F_\omega) =  (\emptyset , \{\mathbf{K}\neg c\}) \).
$(T_{\omega},F_{\omega})$ is not stable, since
$\OBO{\Gamma(\EKAKG \setminus T_{\omega})}$ is unsatisfiable
(with $\Gamma(\EKAKG \setminus T_{\omega})=\{\mathbf{K}a,\mathbf{K}b, \mathbf{K}d, \mathbf{K}\neg c\}$). 
\[
\begin{aligned}
E^{(T_{\omega},F_{\omega})}_{\rK}\!\uparrow\!0 &= (\emptyset,\emptyset),\\
E^{(T_{\omega},F_{\omega})}_{\rK}\!\uparrow\!1 &= ( \emptyset, \{\mathbf{K}\neg c\}) \\
E^{(T_{\omega},F_{\omega})}_{\rK}\!\uparrow\!2 &= ( \emptyset, \{\mathbf{K}\neg c\}) 
\end{aligned}
\]
Since $E^{(T_{\omega},F_{\omega})}_{\rK}\!\uparrow\!3 = E^{(T_{\omega},F_{\omega})}_{\rK}\!\uparrow\!2$, a fixpoint is reached.   \((T_E,F_E) =  ( \emptyset, \{\mathbf{K}\neg c\}) \) is not a stable partition.

We now compute all stable partitions using a guess-and-check approach. 
Let 
\(
U = \EKAKG \setminus (T_E \cup F_E),
\)
so that
\(
U = \{\mathbf{K}\pred{a},\mathbf{K}\pred{b},\mathbf{K}\pred{d}\}.
\)
Since $\lvert U \rvert = 3$, there exist $3^{3} = 27$ partial partitions extending $(T_E, F_E)$, each of which must be checked for stability. Admits two stable partitions
extending $(T_{E},F_{E})$ and satisfying
Prop.~\ref{propos:computingstablepartitions}: \((T_1,F_1) = (\{\mathbf{K}b\},\{\mathbf{K}\neg c, \mathbf{K}a, \mathbf{K}d\})\) and  \((T_2,F_2) = (\{\mathbf{K}b\},\{\mathbf{K}\neg c, \mathbf{K}a\})\). 
The well-founded partition  is
$(T_2,F_2)$, according to
Def.~\ref{def:well-foundedpartition}(\(T_1 = T_2\) and \(F_2 \subseteq F_1\)).

$(T_W,F_W) =  \{\mathbf{K}b\},\{\mathbf{K}\neg c, \mathbf{K}a\}) $

The well-founded model is 
{\small
\[
(M_W,N_W) =
\big(
\{
\{\pred{c},\pred{b}\},
\{\pred{c},\pred{b},\pred{a}\}
\},
\{
\{\pred{c},\pred{b}\},
\{\pred{c},\pred{b},\pred{a}\}
\}
\big).
\]
}
\end{example}
\begin{example}
Consider the $\hybridmknf$ knowledge base
\(
\mathcal{K} = (\rO,\rP),
\)
where
\(
\pi(\rO) = \{c\}
\)
and
\[
\begin{aligned}
\rP = \{\,
r_1:\mathbf{K}a \leftarrow \mathbf{not}\ b. \,
r_2:\mathbf{K}b \leftarrow \mathbf{not}\ a.\,
r_3:\mathbf{K}\neg c \leftarrow \mathbf{K}\neg d,\, \mathbf{K}a. \\
r_4:\mathbf{K}\neg d \leftarrow \mathbf{not}\ h.\,
r_5:\mathbf{K}h \leftarrow \mathbf{not}\ \neg d.
\,\}.
\end{aligned}
\]

\(\EKAKG =  \{\mathbf{K}\pred{a},\mathbf{K}\pred{b},\mathbf{K}\neg \pred{d}, \mathbf{K}\neg \pred{c}, \mathbf{K}\pred{h} \}\).
\[
\begin{aligned}
W_{\rK}\!\uparrow\!0 &= (\emptyset,\emptyset),\\
W_{\rK}\!\uparrow\!1 &= ( \emptyset, \{\mathbf{K}\neg c\}) \\
W_{\rK}\!\uparrow\!2 &= ( \emptyset, \{\mathbf{K}\neg c\})
\end{aligned}
\]

Since $W_{\rK}\!\uparrow\!2 = W_{\rK}\!\uparrow\!1$, a fixpoint is reached.   \((T_\omega,F_\omega) =  (\emptyset , \{\mathbf{K}\neg c\}) \).
$(T_{\omega},F_{\omega})$ is not stable, since
$\OBO{\Gamma(\EKAKG \setminus T_{\omega})}$ is unsatisfiable
(with $\Gamma(\EKAKG \setminus T_{\omega})=\{\mathbf{K}a,\mathbf{K}b, \mathbf{K}d, \mathbf{K}\neg c\}$). 
\[
\begin{aligned}
E^{(T_{\omega},F_{\omega})}_{\rK}\!\uparrow\!0 &= (\emptyset,\emptyset),\\
E^{(T_{\omega},F_{\omega})}_{\rK}\!\uparrow\!1 &= ( \emptyset, \{\mathbf{K}\neg c\}) \\
E^{(T_{\omega},F_{\omega})}_{\rK}\!\uparrow\!2 &= ( \emptyset, \{\mathbf{K}\neg c\}) 
\end{aligned}
\]

Since $E^{(T_{\omega},F_{\omega})}_{\rK}\!\uparrow\!3 = E^{(T_{\omega},F_{\omega})}_{\rK}\!\uparrow\!2$, a fixpoint is reached.   \((T_E,F_E) =  ( \emptyset, \{\mathbf{K}\neg c\}) \) is not a stable partition.

We now compute all stable partitions using a guess-and-check approach. 
Let 
\(
U = \EKAKG \setminus (T_E \cup F_E),
\)
so that
\(
U = \{\mathbf{K}\pred{a},\mathbf{K}\pred{b},\mathbf{K}\pred{d}, \mathbf{K}\pred{h} \}.
\)
\begin{align*}
    (T_1,F_1) = (\{\mathbf{K}\pred{b}\},\{\mathbf{K}\neg c, \mathbf{K}\pred{a}\}) \\
    (T_2,F_2) = (\{\mathbf{K}\pred{h}\},\{\mathbf{K}\neg c, \mathbf{K}\neg \pred{d} \})
\end{align*}

Here there exist no stable partition satisfying Definition \ref{def:well-foundedpartition}, so no well-founded partition.
\end{example}

\section{Proofs}
\begin{proposition}\label{propos:stablepartitionstheoremappneda1}
Let $(M,N)$ be a three-valued MKNF model of a ground $\hybridmknf$ knowledge base $\rKG$, 
and let $(T,F)$ be the partial partition induced by $(M,N)$.
Then
\(
M = \{\, I \mid I \models \OBO{T} \,\}\) and 
\(
N = \{\, I \mid I \models \OBO{\EKAKG \setminus F} \,\}.
\)

\end{proposition}

\begin{proof}

The proof of this proposition is analogous to the proof of the corresponding result in~\cite{Knorr2011LocalCW}. The only additional aspect to be addressed is that the partition $(T,F)$ may contain modal
atoms involving classical negation. We show that this extension does not affect the validity
of the construction. Consider a ground \(\hybridmknf\) knowledge base \(\rKOPG\). Let \((M, N)\) be a three-valued MKNF model of \(\mathrm{\rKG}\), and let \((T, F) \) of  \(\EKAKG\) be the partition induced by \((M, N)\). Let \((M^\prime, N^\prime)\) be the three-valued interpretation pair computed using \ref{propos:stablepartitionstheoremappneda1}.  

We need to prove that \(M = M^\prime\) and \(N = N^\prime\).

First We prove that $M \subseteq M'$. Let $I \in M$. We show that $I \in M'$. From Proposition~\ref{propos:stablepartitionstheoremappneda1}, for every $I' \in M'$, it holds that
$I' \models \OBO{T}$.
By definition, $M'$ is the set of interpretations that satisfy $\OBO{T}$.
Hence, it suffices to show that $I \models \OBO{T}$. Since $(M,N)$ is a three-valued model of $\mathcal{K}_G$, every interpretation $I \in M$ satisfies
$I \models \pi(O)$.
Moreover, $T$ is the set of modal atoms induced by the three-valued model $(M,N)$.
Thus, for every modal atom $\mathbf{K}\xi \in T$, the following condition holds:
\[
\forall I \in M,\;
(I,\langle M,N\rangle,\langle M,N\rangle)\,\mathbf{K}\xi = \mathbf{t}.
\]

Let $\mathbf{K}\neg A \in T$.
By the semantics of modal atoms, $\mathbf{K}\neg A = \mathbf{t}$ holds if and only if
$A$ is false in every interpretation $I \in M$.
Hence, for all $I \in M$, we have $I \models \neg A$, and therefore $I \not\models A$.
Equivalently, \(
\forall I \in M,\; A \notin I.
\) Consequently,
\[
\forall I \in M,\;
(I,\langle M,N\rangle,\langle M,N\rangle)\,\mathbf{K}\neg A = \mathbf{t}.
\]

Since this argument applies to every modal atom in $T$, it follows that
$I \models \OBO{T}$.
Therefore, $I \in M'$, and we conclude that $M \subseteq M'$. We prove that $N \subseteq N'$. Let $I \in N$. We show that $I \in N'$. By definition, $I \in N'$ if and only if $I \models \OBO{\EKAKG \setminus F}$.
Let \(
U = \EKAKG \setminus (T \cup F).
\)
Thus, it suffices to show that
\[
I \models \pi(O)\ \cup\ \{\xi \mid \mathbf{K}\xi \in T\}\ \cup\ \{\xi \mid \mathbf{K}\xi \in U\}.
\]

Since $N \subseteq M$, we have $I \in M$. Hence $I \models \OBO{T}$, and therefore
\(
I \models \pi(O)\ \cup\ \{\xi \mid \mathbf{K}\xi \in T\}.
\) It remains to show that $I \models \{\xi \mid \mathbf{K}\xi \in U\}$.
Let $\mathbf{K}\xi \in U$ be arbitrary. Since $(M,N)$ is a three-valued model of
$\mathcal{K}_G$, we have that
\(
\forall I \in M:\ (I,\langle M,N\rangle,\langle M,N\rangle)\,\mathbf{K}\xi = \mathbf{u}
\ \text{only if}\ 
\xi \notin J_{1}\ \text{for some}\ J_{1}\in M
\ \text{and}\ 
\xi \in J_{2}\ \text{for all}\ J_{2}\in N.
\)
In particular, $\xi \in J$ for all $J \in N$, and since $I \in N$, it follows that
$I \models \xi$. As $\mathbf{K}\xi \in U$ was arbitrary, we conclude that
\(
I \models \{\xi \mid \mathbf{K}\xi \in U\}.
\) Similarly, for classical negation, since $(M,N)$ is a three-valued model of
$\mathcal{K}_G$, we have that
\(
\forall I \in M:\ (I,\langle M,N\rangle,\langle M,N\rangle)\,\mathbf{K}\neg \xi = \mathbf{u}
\ \text{only if}\ 
\xi \in J_{1}\ \text{for some}\ J_{1}\in M
\ \text{and}\ 
\xi \notin J_{2}\ \text{for all}\ J_{2}\in N.
\) Therefore,
\(
I \models \pi(O)\ \cup\ \{\xi \mid \mathbf{K}\xi \in T\}\ \cup\ \{\xi \mid \mathbf{K}\xi \in U\}.
\)
Consequently, $I \models \OBO{\EKAKG \setminus F}$, and hence $I \in N'$.
Thus, $N \subseteq N'$.

Now to prove \(M^\prime = M\) and \(N^\prime = N\).  We know that \((M, N)\) is a three-valued model of \(\rKG\), so \(\forall I \in M:\ (I,\langle M,N\rangle,\langle M,N\rangle)  \rKG = t\) implies \(\forall I \in M :\ (I, \langle M, N \rangle, \langle M, N \rangle) \textbf{K} \pi(O) \wedge \pi(\rPG) = t\). Now, Proposition \ref{propos:stablepartitionstheoremappneda1}, we have \(\forall I \in M :\ (I, \langle M^\prime, N^\prime \rangle, \langle M, N \rangle)  \textbf{K} \pi(O) = t\).  It remains to show that\(\forall I \in M :\ (I, \langle M^\prime, N^\prime \rangle, \langle M, N \rangle)  \textbf{K} \pi(\rPG) = t\). Recall that \(T \subseteq \EKAKG\) are the sets of  true positive literals and classical negated literals, respectively, with respect to \(\rPG\).

\begin{enumerate}
    \item From Definition~\ref{extendedKAthreeinducedparition}, for every $\mathbf{K}\xi \in T$ and the three-valued MKNF model $(M,N)$, we have
\(
\forall I \in M:\ (I,\langle M,N\rangle,\langle M,N\rangle)(\mathbf{K}\xi)=\mathbf{t}.
\)
Moreover, for every $I' \in M'$, since $I' \models \OBO{T}$ and $\mathbf{K}\xi \in T$, it follows that
\(
\forall I' \in M':\ (I',\langle M',N'\rangle,\langle M,N\rangle)(\mathbf{K}\xi)=\mathbf{t}.
\)

    \item Similarly, for every $\mathbf{K}\neg \xi \in T$ and the three-valued MKNF model
$(M,N)$, we have
\(
\forall I \in M:\ (I,\langle M,N\rangle,\langle M,N\rangle)(\mathbf{K}\neg \xi)
= \mathbf{t}.
\)
By the semantics of $\mathbf{K}$ with classical negation, this implies that
$\xi$ is false in every interpretation $I \in M$.

The same holds for the three-valued interpretation $(M',N')$.
Indeed, by Proposition~\ref{propos:stablepartitionstheoremappneda1},
let $M'$ be the set of interpretations $I'$ such that $I' \models OB_{O,T}$.
Then, for every $I' \in M'$, it holds that $I' \models \neg \xi$, and hence
$I' \not\models \xi$.
Equivalently,
\(
\forall I' \in M',\ \xi \notin I'.
\)
Consequently,
\(
\forall I' \in M':\
(I',\langle M',N'\rangle,\langle M,N\rangle)(\mathbf{K}\neg \xi)
= \mathbf{t}.
\)

    \item Let $\mathbf{K}\xi \in F$ and consider the three-valued MKNF model $(M,N)$.
 We have
 
\(
(I,\langle M,N\rangle,\langle M,N\rangle)(\mathbf{K}\xi) = \mathbf{f}.
\)
This means that there exists an interpretation $J \in N$ such that
$J \not\models \xi$.

Since $N \subseteq N'$, it follows that there exists an interpretation
$J' \in N'$ such that $J' \not\models \xi$.
Therefore,
\(
(I',\langle M',N'\rangle,\langle M,N\rangle)(\mathbf{K}\xi) = \mathbf{f},
\)
for every $I' \in M'$.

    \item Let $\mathbf{K}\neg \xi \in F$ and consider the three-valued MKNF model $(M,N)$.
By the semantics of $\mathbf{K}$, we have

\(
(I,\langle M,N\rangle,\langle M,N\rangle)(\mathbf{K}\neg \xi) = \mathbf{f}.
\)
This means that there exists an interpretation $J \in N$ such that
$J \models \xi$.

Since $N \subseteq N'$, it follows that there exists an interpretation
$J' \in N'$ such that $J' \models \xi$.
Therefore,
\(
(I',\langle M',N'\rangle,\langle M,N\rangle)(\mathbf{K}\neg \xi) = \mathbf{f},
\)
for every $I' \in M'$.

    \item For every $\textbf{K}\xi \in U$, and the three-valued MKNF model \((M, N)\), we have
    
    \(\forall I \in M :\ (I, \langle M, N \rangle, \langle M, N \rangle)  \textbf{K}\xi = u\),  this holds if and only if there exists an
interpretation $I_{1}\in M$ such that $I_{1}\not\models \xi$, and for all
interpretations $J\in N$, it holds that $J\models \xi$.

Since $M\subseteq M'$, there exists an interpretation $I'_{1}\in M'$ such that
$I'_{1}\not\models \xi$.
Moreover, by definition of $N'$, every interpretation $J'\in N'$ satisfies
$\pi(O)\cup T\cup U$.
In particular, for every $J'\in N'$, we have $J'\models \xi$. Therefore, there exists an interpretation in $M'$ in which $\xi$ is false,
and $\xi$ is true in all interpretations in $N'$.
Hence,
\(
(I',\langle M',N'\rangle,\langle M,N\rangle)(\mathbf{K}\xi)=\mathbf{u},
\)
for every $I'\in M'$.

    \item For every  $\textbf{K}\neg\xi \in U$,and  the three-valued MKNF model \((M, N)\), we have
    
    \(\forall I \in M :\ (I, \langle M^\prime, N^\prime \rangle, \langle M, N \rangle) \textbf{K}\xi = u\).  This happen when some \(I \in M\) where \(\xi \in I\) and all \(J \in N\) where \(\xi \notin I\).  Since \(M \subseteq  M^\prime\) so there exist at-least one interpretation \(I^\prime \in M^\prime\) where \(\xi \in I^\prime\).    Moreover, by definition of $N'$, every interpretation $J'\in N'$ satisfies
$\pi(O)\cup T\cup U$.
In particular, for every $J'\in N'$, we have $J'\\notmodels \xi$ if \(\textbf{K}\neg\xi \in U\) Therefore, there exists an interpretation in $M'$ in which $\xi$ is true,
and $\xi$ is false in all interpretations in $N'$.
Hence,
\(
(I',\langle M',N'\rangle,\langle M',N'\rangle)(\mathbf{K}\neg \xi)=\mathbf{u},
\)
for every $I'\in M'$.
\end{enumerate}
\noindent
Hence, we conclude \(
\forall I^\prime \in M^\prime :\ (I, \langle M^\prime, N^\prime \rangle, \langle M, N \rangle)  (K\pi(O) \wedge \pi(\rPG)) = t,
\)
which implies that the evaluations under $(M, N)$ and $(M^\prime, N^\prime)$
coincide for every ground atom of the \(\hybridmknf\) knowledge base.
Consequently, both pairs $(M, N)$ and $(M^\prime, N^\prime)$ represent the same model.
\end{proof}
\begin{lemma}\label{lem:3valued-evala1}
Let $(T,F)$ be a partial partition of $\EKAKG$, and let $(M,N)$ be the MKNF
interpretation induced by $(T,F)$. For every $\xi$ such that
$\mathbf{K}\xi \in \EKAKG$,

\vspace{-.02em}
{\small
\begin{minipage}{.48\linewidth}
\[
\mathbf{K}\xi[T,F]=
\begin{cases}
\mathbf{t} &  \OBO{T}\models\xi,\\
\mathbf{f} & \OBO{\EKAKG\setminus F}\not\models\xi,\\
\mathbf{u} & \text{otherwise}.
\end{cases}
\]
\end{minipage}\hspace{.5em}
\begin{minipage}{.48\linewidth}
\[
\mathbf{not}\ \xi[T,F]=
\begin{cases}
\mathbf{t} & \OBO{\EKAKG\setminus F} \not\models\xi,\\
\mathbf{f} &\OBO{T}\models\xi ,\\
\mathbf{u} & \text{otherwise}.
\end{cases}
\]
\end{minipage}
}

\end{lemma}

\begin{proof}
Let $(T,F)$ be a partial partition of $\EKAKG$ and let $(M,N)$ be the MKNF
interpretation induced by $(T,F)$, with
$M=\{I \mid I\models \OBO{T}\}$ and
$N=\{J \mid J\models \OBO{\EKAKG\setminus F}\}$.

(i) $\forall I \in M$,  $(I,(M,N),(M,N))(\bK\xi)=\mathbf{t}$ holds precisely when
$\forall I\in M,\ I\models\xi$.
Since $M=\{I\mid I\models \OBO{T}\}$, this means that for all $I\in M$,
$I\models \OBO{T}$ and $I\models\xi$.
This is equivalent to $\OBO{T}\models\xi$.

(ii)  $\forall I \in M$, $(I,(M,N),(M,N))(\bK\xi)=\mathbf{f}$ holds exactly when
$\exists J\in N,\ J\not\models\xi$.
Since $N=\{J\mid J\models \OBO{\EKAKG\setminus F}\}$, this condition is equivalent
to $\OBO{\EKAKG\setminus F}\not\models\xi$.

(iii)  $\forall I \in M$, $(I,(M,N),(M,N))(\bK\xi)=\mathbf{u}$ holds exactly when
$\forall J\in N,\ J\models\xi$ and $\exists I\in M,\ I\not\models\xi$.
Since $N=\{J\mid J\models \OBO{\EKAKG\setminus F}\}$, the first condition is
equivalent to $\OBO{\EKAKG\setminus F}\models\xi$.
Since $M=\{I\mid I\models \OBO{T}\}$, the second condition is equivalent to
$\OBO{T}\not\models\xi$.
Since the same $K\xi$ satisfies $K\xi \in \EKAKG \setminus F$ and $K\xi \notin T$,
both conditions hold.

(iv) $\forall I \in M$,  $(I,(M,N),(M,N))(\bnot \xi)=\mathbf{t}$ holds precisely when
$\exists J\in N,\ J\not\models\xi$.
this condition is equivalent
to $\OBO{\EKAKG\setminus F}\not\models\xi$.

(v)  $\forall I \in M$, $(I,(M,N),(M,N))(\bnot\xi)=\mathbf{f}$ holds exactly when
$\forall I\in M,\ I\models\xi$.
Since $M=\{I\mid I\models \OBO{T}\}$, this means that for all $I\in M$,
$I\models \OBO{T}$ and $I\models\xi$.
This is equivalent to $\OBO{T}\models\xi$.

(vi) $\forall I \in M$,  $(I,(M,N),(M,N))(\bnot\xi)=\mathbf{u}$ holds exactly when
$\forall J\in N,\ J\models\xi$ and $\exists I\in M,\ I\not\models\xi$.
Since $N=\{J\mid J\models \OBO{\EKAKG\setminus F}\}$, the first condition is
equivalent to $\OBO{\EKAKG\setminus F}\models\xi$.
Since $M=\{I\mid I\models \OBO{T}\}$, the second condition is equivalent to
$\OBO{T}\not\models\xi$.
Since the same $K\xi$ satisfies $K\xi \in \EKAKG \setminus F$ and $K\xi \notin T$,
both conditions hold.

This completes the proof.
\end{proof}

\begin{lemma}\label{lemma:proofstbaleparitiomna2}
Let $(T,F)$ be a partial partition of $\EKAKG$ and let $(M,N)$ be the MKNF interpretation
induced by $(T,F)$. Then, for every rule $r\in\rPG$,
\(
\big( \forall I \in M,\ (I,\langle M,N\rangle,\langle M,N\rangle)(\pi(r))=\mathbf{t} \big)
\ \text{iff}\ 
\big( h(r)[T,F]\ge b(r)[T,F] \big),
\)
where
\(
b(r)[T,F]=\min\{\ell[T,F]\mid \ell\in b^+(r)\cup b^-(r)\}.
\)
\end{lemma}

\begin{proof}
Fix $r\in\rPG$. By definition, $\pi(r)$ is the formula
\[
\pi(r)\;=\; \bigl(b^+(r)\wedge b^-(r)\bigr)\supset h(r).
\]

Let $(T,F)$ be the partial partition of $\EKAKG$ and let $(M,N)$ be the MKNF
interpretation induced by $(T,F)$. Consider the MKNF structure
$(\mathcal{I},\mathcal{M},\mathcal{N})$ with $\mathcal{M}=\langle M,N\rangle$
and $\mathcal{N}=\langle M,N\rangle$.

Write $b^+(r)\wedge b^-(r)$ as $\ell_1\wedge\cdots\wedge \ell_k$, where
$\{\ell_1,\dots,\ell_k\}=b^+(r)\cup b^-(r)$.
By Equation~(5) in Fig.~\ref{eq:mknfevalaution},
\[
(\mathcal{I}, \mathcal{M}, \mathcal{N})(\ell_1\wedge\cdots\wedge \ell_k)
= \min\{\,(\mathcal{I}, \mathcal{M}, \mathcal{N})(\ell_i)\mid 1\le i\le k \,\}.
\]
By Lemma~\ref{lem:3valued-evala1}, for every literal $\ell$,
\[
(\mathcal{I}, \mathcal{M}, \mathcal{N})(\ell)=\ell[T,F].
\]
Therefore,
\[
(\mathcal{I}, \mathcal{M}, \mathcal{N})(b^+(r)\wedge b^-(r))
= \min\{\, \ell[T,F]\mid \ell\in b^+(r)\cup b^-(r)\,\}
= b(r)[T,F].
\]

Now apply Equation~(6) to the implication $\pi(r)$.
Taking $\varphi_1:=b^+(r)\wedge b^-(r)$ and $\varphi_2:=h(r)$, Equation~(6)
yields
\[
(\mathcal{I}, \mathcal{M}, \mathcal{N})(\pi(r))=\mathbf{t}
\quad\text{iff}\quad
(\mathcal{I}, \mathcal{M}, \mathcal{N})(h(r))\ge
(\mathcal{I}, \mathcal{M}, \mathcal{N})(b^+(r)\wedge b^-(r)).
\]
Since $h(r)$ is a modal atom, Lemma~\ref{lem:3valued-evala1} gives
\[
(\mathcal{I}, \mathcal{M}, \mathcal{N})(h(r))=h(r)[T,F].
\]
Combining the above equalities, we obtain
\[
(\mathcal{I}, \mathcal{M}, \mathcal{N})(\pi(r))=\mathbf{t}
\quad\text{iff}\quad
h(r)[T,F]\ge b(r)[T,F],
\]
as required.
\end{proof}

\begin{theorem}\label{thereoe:stbaleparititoa3}
Let $\rKOPG$ be a ground $\hybridmknf$ knowledge base. Let $(T,F)$ be a partial partition of $\EKAKG$.
$(T,F)$ is stable
if and only if the three-valued MKNF interpretation $(M,N)$ induced by $(T,F)$
is a three-valued MKNF model of $\rKG$.
\end{theorem}

\begin{proof}
To prove the theorem, it suffices to establish the following three conditions:

\begin{enumerate}[leftmargin=*, itemsep=1.5pt]
    \item Three-valued interpretation.  
    We show that $(M,N)$ is a three-valued interpretation, that is,
    \(
        \emptyset \subseteq N \subseteq M.
    \)

    \item Satisfaction condition.  
    We show that $(M,N)$ satisfies the knowledge base $\rKG$, namely,
    \(
        \rKG \models (M,N),
\)
    which is equivalent to
    \(
        \forall I \in M : (I,\langle M,N\rangle ,\langle M,N\rangle )\bigl(\pi(\rKG)\bigr) = \mathbf{t}.
\)

    \item Maximality criterion. 
    We show that there exist no interpretations $(M',N')$ such that
    \(
        M \subseteq M' \) and  \( N \subseteq N'\)
 with at least one inclusion being proper, and
    \(
        (I,\langle M,N\rangle ,\langle M,N\rangle )\bigl(\pi(\rKG)\bigr) \neq \mathbf{t}.
  \)
\end{enumerate}

\paragraph{Three-valued interpretation: Inclusion \(\emptyset \subseteq N \subseteq M\).}
From Definition~\ref{def:semanticsofstablepartition},
$\OBO{\EKAKG \setminus F}$ is satisfiable (condition (i)). Therefore, there exists an
interpretation $J$ such that $J \models \OBO{\EKAKG \setminus F}$,
which implies that $N \neq \emptyset$, and hence $\emptyset \subseteq N$. We now show that $N \subseteq M$. Let $J \in N$. By definition of $N$,
$J \models \OBO{\EKAKG \setminus F}$. Since $T \subseteq \EKAKG \setminus F$,
every model of $\OBO{\EKAKG \setminus F}$ is also a model of $\OBO{T}$.
Thus, $J \models \OBO{T}$, and therefore $J \in M$. The converse inclusion does not hold in general. Indeed, there may exist
an interpretation $I \in M$ such that $I \models \OBO{T}$ but
$I \not\models \OBO{\EKAKG \setminus F}$, since
$\OBO{\EKAKG \setminus F}$ may contain axioms not present in $\OBO{T}$.
Consequently, $M \nsubseteq N$ in general.

\paragraph{Satisfaction condition: $(M,N) \models \rKG$.}
By definition, we have
\(
\pi(\rKG)= \mathbf{K}\pi(O) \wedge \pi(\rPG).
\)
Thus, to show that $(M,N) \models \rKG$, it suffices to prove that
\(
(M,N) \models \mathbf{K}\pi(O) \quad \text{and} \quad (M,N) \models \pi(\rPG).
\) By Proposition~\ref{propos:stablepartitionstheorem}, $\pi(O)$ is satisfied by
every interpretation $I \in M$.
Hence, $(M,N) \models \mathbf{K}\pi(O)$. It remains to show that $(M,N) \models \pi(\rPG)$.
From Definition~\ref{def:semanticsofstablepartition}, condition~(ii),
for every rule $\pi(r) \in \pi(\rPG)$ and every interpretation $I \in M$,
we have
\(
\forall I \in M : (I,\langle M,N\rangle ,\langle M,N\rangle )(\pi(r)) = \mathbf{t}.
\)
Therefore,
\(
\forall I \in M : (I,\langle M,N\rangle ,\langle M,N\rangle )(\pi(\rPG)) = \mathbf{t}.
\)
Consequently, $(M,N) \models \rO$ and $(M,N) \models \rPG$
 which implies $(M,N) \models \rKG$.

\paragraph{Maximality criterion.}
We show that there exists no pair $(M',N')$ such that
$M \subseteq M'$ and $N \subseteq N'$, with at least one inclusion being proper,
and $(M',N') \models \rKG$.
Assume towards a contradiction that there exists a three-valued MKNF interpretation $(M',N')$
such that $M\subseteq M'$ and $N\subseteq N'$, with at least one inclusion being proper, and
\((M',N')\) is three-valihed MKNF model  that is, $(M',N')\models\rKG$.
Let $(T',F')$ be the partial partition induced by $(M',N')$.
By Proposition~\ref{propos:stablepartitionstheorem},
$\OBO{T'} \subseteq \OBO{T}$ and
$\OBO{\EKAKG \setminus F'} \subseteq \OBO{\EKAKG \setminus F}$,
since $M\subseteq M'$ and $N\subseteq N'$, and  monotonicity of $\OBO{\cdot}$,
$T'\subseteq T$ and $F\subseteq F'$, with at least one inclusion being proper.
Moreover, since $(M',N')\models\rKG$, the condition of Definition~\ref{def:semanticsofstablepartition}(ii)
holds for $(T',F')$, contradicting Definition~\ref{def:semanticsofstablepartition}(iii).
Hence, no such $(M',N')$ exists, and $(M,N)$ is maximal.
Therefore, $(M,N)$ is a three-valued MKNF model of $\rKOPG$.

Conversely, let $(M,N)$ be a three-valued MKNF model of $\rKG$, and let $(T,F)$ be the partial partition induced by $(M,N)$.
We show that $(T,F)$ is a stable partition according to Definition~\ref{def:semanticsofstablepartition}.
\paragraph{Condition (i).}
Since $(M,N)$ is a three-valued MKNF model, Proposition~\ref{propos:stablepartitionstheoremappneda1} guarantees that
\(
\OBO{\EKAKG \setminus F}
\)
is satisfiable. Hence, Condition~(i) holds.

\paragraph{Condition (ii).}
Condition~(ii) consists of two parts.

\subparagraph{Condition (ii.1).}
Let $\mathbf{K}\xi \in \EKAKG$.
If $\OBO{T} \models \xi$, then by Proposition~\ref{propos:stablepartitionstheoremappneda1}, every interpretation $I \in M$ satisfies $\xi$. Hence, $\mathbf{K}\xi \in T$. If $\OBO{\EKAKG \setminus F} \not\models \xi$, then there exists an interpretation compatible with
$\EKAKG \setminus F$ that falsifies $\xi$. this implies
$\mathbf{K}\xi \in F$. Therefore, Condition~(ii.1) is satisfied.

\subparagraph{Condition (ii.2).}
For every rule $r \in \rPG$ and every interpretation $I \in M$, we have
\(
\forall I \in M : (I,\langle M,N\rangle ,\langle M,N\rangle )(\pi(r)) = t,
\)
since $(M,N)$ is a three-valued MKNF model,  so \((M,N) \models \rKG \) hence \(\forall I \in M : (I,\langle M,N\rangle ,\langle M,N\rangle )(\pi(\rPG)) = \mathbf{t},\).
Thus, Condition~(ii.2) holds.

\paragraph{Condition (iii).}

Since $(M,N)$ is a three-valued MKNF model, there exists no three-valued MKNF interpretation $(M',N')$
such that
\(
M \subseteq M'  \ \ \  \text{and} \ \ \  N \subseteq N',
\)
for which $\rKG \models (M',N')$. \(M \subseteq M' \ \text{and} \  N \subseteq N',\) atleast on of the inclusion being proper Let $(T',F')$ be a partial partition induced by  $(M',N')$ such that $T' \subseteq T$ and $F \subseteq F'$, with at least one
of these inclusions being proper.  

Consequently, at least one of the following holds:
\begin{itemize}
    \item there exists $\mathbf{K}\xi \in \EKAKG \setminus T'$ such that $\OBO{T'} \models \xi$; or
    \item there exists $\mathbf{K}\xi \in \EKAKG \setminus F'$ such that
    $\OBO{\EKAKG \setminus F'} \not\models \xi$; or
    \item there exists a rule $r \in \rPG$ such that for all $I \in M'$,
    \(
  (I,\langle M^\prime,N^\prime \rangle ,\langle M,N\rangle )(\pi(r)) = f.
  \)
\end{itemize}
Hence, Condition~(iii) is satisfied. All conditions of Definition~\ref{def:semanticsofstablepartition} are satisfied. Therefore, the induced
partition $(T,F)$ by a three-valued MKNF model $(M,N)$ is a stable partition.
\end{proof}

\begin{proposition}\label{propos:computingstablepartitionss3}
Let $\rKOPG$ be a ground $\hybridmknf$ knowledge base.
$(T,F)$ is a stable partition of $\EKAKG$ if and only if (i) $T =  \Gamma(F) $, (ii) $F = Fa(\EKAKG \setminus T)$ and (iii) $\OBO{\Gamma(\EKAKG \setminus T)}$ is satisfiable.
\end{proposition}

\begin{proof}
Let $\rKG = (\rO,\rPG)$ be a ground hybrid MKNF knowledge base, and let
$(T,F)$ be a partial partition of $\EKAKG$.
By Proposition~\ref{propos:stablepartitionstheorem}, the partition $(T,F)$
induces a three-valued MKNF interpretation $(M,N)$.
To show that $(M,N)$ is a three-valued MKNF model of $\rKG$, it suffices to
establish that $(T,F)$ is a stable partition.

According to Definition~\ref{def:semanticsofstablepartition}, this requires
showing three conditions. Here we prove that if a partial partition $(T,F)$
satisfies $T = \Gamma(F)$, $F = Fa(\EKAKG \setminus T)$, and
$\OBO{\Gamma(\EKAKG \setminus T)}$ is satisfiable, then it satisfies all the
conditions of a stable partition.

\emph{(i)} $T = \Gamma(F)$, where $\Gamma(F)$ is the least fixpoint of
$T^F_{\mathcal K_G}$, and
\begin{align}\scriptsize
T^F_{\mathcal K_G}(X)
= {} &
\{\, h(r) \mid r\in\mathcal P_G,\;
b^+(r)\subseteq X,\;
\mathbf K(b^-(r))\subseteq F \,\}
\nonumber\\
&\cup
\{\, \mathbf K\xi \mid \mathbf K\xi\in\EKAKG,\;
\OBO{X} \models \xi \,\},
\label{eq:truea1}
\end{align}
$\Gamma(F) =  lfp(T^F_{\mathcal K_G})$. Fix an arbitrary rule $r\in\mathcal P_G$ and write
\[
\pi(r)=b^+(r)\wedge b^-(r)\supset h(r).
\]
Let $(X_n)_{n\ge 0}$ be the sequence defined by
\[
X_0= \emptyset
 \ \ \  \text{and} \ \ \ 
X_{n+1}=T^F_{\mathcal K_G}(X_n)\ \ (n\ge 0).
\]
Since $T^F_{\mathcal K_G}$ is monotone and $\EKAKG$ is finite, the sequence
stabilizes at some $k$, that is, $X_k=X_{k+1}$, and hence
\[
\Gamma(F)= lfp(T^F_{\mathcal K_G})=\bigcup_{n\ge 0} X_n .
\]

In particular, $X_1=T^F_{\mathcal K_G}(X_0)$ contains exactly:
(i) all heads $h(r)$ of rules $r\in\mathcal P_G$ whose positive body is empty,
that is, $b^+(r)=\varnothing$, and whose negative modal conditions satisfy
$\mathbf K(b^-(r))\subseteq F$; and
(ii) all modal atoms $\mathbf K\xi\in\EKAKG$ such that $\OBO{\varnothing}\models\xi$.
More generally, for every $n\ge 0$, $X_{n+1}$ contains the heads of all rules
$r$ such that $b^+(r)\subseteq X_n$ and $\mathbf K(b^-(r))\subseteq F$, together
with all $\mathbf K\xi$ entailed by $\OBO{X_n}$.

$\Gamma(F)$ contains $h(r)$ for all rules $r \in \rPG$ such that
$b(r)[T,F]$ is true. \(b^+(r)\subseteq X,\; \) and \(\mathbf K(b^-(r))\subseteq F\). \(\{\, \mathbf K\xi \mid \mathbf K\xi\in\EKAKG,\;
\OBO{X} \models \xi \,\}\) will ensure that  for every $\mathbf{K}\xi\in\EKAKG$,
$\OBO{T}\models\xi$ implies $\mathbf{K}\xi\in T$. Thus, $T = \Gamma(F)$ shows that all derivable modal
atoms with respect to $F$ are in $T$.
\begin{align}\scriptsize
TU^{\EKAKG \setminus T}_{\mathcal K_G}(X)
= {} &
\{\, h(r) \mid r\in\mathcal P_G,\;
b^+(r)\subseteq X,\;
\mathbf K(b^-(r))\subseteq (\EKAKG \setminus T)
\; \text{ and } \OBO{X} \not\models \overline{h(r)}\,\}
\nonumber\\
&\cup
\{\, \mathbf K\xi \mid \mathbf K\xi\in\EKAKG,\;
\OBO{X} \models \xi \,\},
\label{eq:true}
\end{align}

Let $(Y_n)_{n\ge 0}$ be the sequence defined by
\[
Y_0=\varnothing
\  \ \ \text{and} \ \ \ 
Y_{n+1}=TU^{\EKAKG \setminus T}_{\mathcal K_G}(Y_n)\ \ (n\ge 0).
\]
Since $\EKAKG$ is finite, the sequence stabilizes at some $k$, that is,
$Y_k=Y_{k+1}$, and we denote
\[
\Gamma'(\EKAKG \setminus T)= lfp\!\left(TU^{\EKAKG \setminus T}_{\mathcal K_G}\right)=\bigcup_{n\ge 0} Y_n .
\]

In particular, $Y_1=TU^{\EKAKG \setminus T}_{\mathcal K_G}(Y_0)$ contains exactly:
(i) all heads $h(r)$ of rules $r\in\mathcal P_G$ such that $b^+(r)=\varnothing$,
$\mathbf K(b^-(r))\subseteq (\EKAKG\setminus T)$, and
$\OBO{\varnothing}\not\models \overline{h(r)}$; and
(ii) all modal atoms $\mathbf K\xi\in\EKAKG$ such that
$\OBO{\varnothing}\models\xi$.

More generally, for every $n\ge 0$, $Y_{n+1}$ contains the heads of all rules
$r$ such that $b^+(r)\subseteq Y_n$,
$\mathbf K(b^-(r))\subseteq (\EKAKG\setminus T)$, and
$\OBO{Y_n}\not\models \overline{h(r)}$, together with all modal atoms
$\mathbf K\xi$ entailed by $\OBO{Y_n}$.

Let $\Gamma'(S)$ be the least fixpoint of $TU^S_{\mathcal K_G}$, which
contains $h(r)$ for all rules $r \in \rPG$ such that $b(r)[T,F]$ is true or
possibly true (undefined). Modal atoms that cannot be derived in this way
are considered false and belong to $Fa({\EKAKG \setminus T})$, defined as
$\EKAKG \setminus \Gamma'({\EKAKG \setminus T})$. This also satisfy $\OBO{\EKAKG\setminus F}\not\models\xi$ implies $\mathbf{K}\xi\in F$;

The only potential violation arises when there exists a rule \(r \in \rPG\) such that
$b(r)[T,F]=\mathbf{u}$ and $\overline{h(r)}[T,F]=\mathbf{t}$; so  $h(r)[T,F]=\mathbf{f}$  hence r will not satisfied withrespect to $(T,F)$ however, this
case is excluded by the satisfiability of
$\OBO{\Gamma(\EKAKG\setminus T)}$. \(\Gamma(\EKAKG\setminus T)\).  Moreover, the satisfiability of $\OBO{\EKAKG \setminus F}$ ensures that the
partition $(T,F)$ induces a three-valued MKNF model $(M,N)$ according to
Proposition~\ref{propos:stablepartitionstheorem}.

Hence, conditions~i.1 and~i.2 of
Definition~\ref{def:semanticsofstablepartition} are satisfied. We must show that for every partial partition $(T',F')$ such that
\(
T' \subseteq T  \ \  \text{and}  \ \ F \subseteq F',
\)
with at least one inclusion being proper, there exists a rule
$r \in \rPG$ such that the MKNF interpretation $(M',N')$ induced by
$(T',F')$ satisfies
\[
(I,(M',N'),(M,N))(\pi(r)) = \mathbf{f}.
\]

The maximality proof follows from $T = \Gamma(F)$ and
$F = Fa(\EKAKG \setminus T)$, since the least fixpoint construction implies
that there is no smaller set that still satisfies the required conditions.
\end{proof}

\begin{theorem}[Well-founded Model]\label{theorem:wfmodela1}
Let $\rKG$ be a ground $\hybridmknf$ knowledge base.
If $(T_W,F_W)$ is the well-founded partition of $\rKG$,
then the MKNF interpretation pair $(M_W,N_W)$ induced by $(T_W,F_W)$
is the well-founded model of $\rKG$.
\end{theorem}

\begin{proof}
Let $(T_W,F_W)$ be the well-founded partition of $\rKG$, and let $(T,F)$ be any stable partition of $\rKG$.
By Definition~\ref{def:well-foundedpartition}, we have $T_W \subseteq T$ and $F_W \subseteq F$.

Since $T_W \subseteq T$, every interpretation that satisfies $\OBO{T}$ also satisfies $\OBO{T_W}$.
Thus,
\[
\{ I \mid I \models \OBO{T} \} \subseteq \{ I \mid I \models \OBO{T_W} \}.
\]
this implies $M \subseteq M_W$ (From Proposition \ref{propos:stablepartitionstheoremappneda1}).

Similarly, from $F_W \subseteq F$ it follows that $\EKAKG \setminus F \subseteq \EKAKG \setminus F_W$.
Hence,
\[
\{ I \mid I \models \OBO{\EKAKG \setminus F_W} \}
\subseteq
\{ I \mid I \models \OBO{\EKAKG \setminus F} \},
\]
which implies $N_W \subseteq N$. Therefore, for every three-valued MKNF model $(M,N)$ induced by a stable partition of $\rKG$, we have
$M \subseteq M_W$ and $N_W \subseteq N$.
That is, $(M_W,N_W)$ is minimal with respect to the information ordering
defined by
\(
(M_W,N_W) \preceq_k (M,N) \ \  \text{iff} \ \ \ M \subseteq M_W \text{ and } N_W \subseteq N,
\)
such that for all three-valuhed MKNF model \((M,N)\)as introduced in~\cite{Knorr2011LocalCW}. Consequently, $(M_W,N_W)$ is the well-founded MKNF model of $\rKG$.
\end{proof}

\begin{proposition}\label{prop:trueandfalseineverythreevaluedwellfoudnedoperaptra1}
Let \(\rKOPG\) be a ground \(\hybridmknf\) knowledge base and
\((T_\omega,F_\omega)\) a fixpoint of the well-founded operator \(W_{\rKG}\).
For any modal atom \(\mathbf{K}H \in \EKAKG\), with \(H\) of the form \(\xi\) or \(\neg\xi\),
and for every stable partition \((T,F)\) of \(\rKG\), the following holds:
\(\mathbf{K}H \in T_\omega\) implies \(\mathbf{K}H \in T\), and
\(\mathbf{K}H \in F_\omega\) implies \(\mathbf{K}H \in F\).
\end{proposition}

\begin{proof}
Let $W_{\mathcal K_G}\uparrow 0=(\emptyset,\emptyset)$ and
$W_{\mathcal K_G}\uparrow (i+1)=W_{\mathcal K_G}(W_{\rKG}\uparrow i)$ for all $i\ge 0$, and write
$W_{\mathcal K_G}\uparrow i=(T_i,F_i)$. Let  an arbitrary stable partition of $\rKG$ and denote it by $(T,F)$.

We prove by induction on $i\ge 0$ the following property $(\mathcal I_i)$:
For every ground modal atom $\mathbf K H \in \EKAKG$, if $\mathbf K H \in T_i$ then
$\mathbf K H \in T$, and if $\mathbf K H \in F_i$ then $\mathbf K H \in F$.

\paragraph{Base case ($i=0$).}
Since $W_{\mathcal K_G}\uparrow 0=(\emptyset,\emptyset)$, $(\mathcal I_0)$ holds
trivially.

\paragraph{Induction step.}
Assume that $(\mathcal I_i)$ holds for $W_{\mathcal K_G}\uparrow i=(T_i,F_i)$.
We show that $(\mathcal I_{i+1})$ holds for
\[
W_{\mathcal K_G}\uparrow (i+1)=(T_{i+1},F_{i+1})=W_{\mathcal K_G}(T_i,F_i).
\]

\medskip
\noindent\emph{Induction hypothesis(IH1).}
For every ground modal atom $\mathbf K H \in \EKAKG$, if $\mathbf K H \in T_i$ then
$\mathbf K H \in T$, and if $\mathbf K H \in F_i$ then $\mathbf K H \in F$.

By Definition~\ref{def:well-foundedoperator},
\[
(T_{i+1},F_{i+1}) = W_{\rKG}(T_i,F_i)
=
\bigl(
T^{F_i}_{\rKG}(T_i),
F_{\rKG}(T_i,F_i)
\bigr).
\]

We first prove the first part of $(\mathcal I_{i+1})$, namely that for every ground
modal atom $\mathbf K H \in \EKAKG$, if $\mathbf K H \in T_{i+1}$ then
$\mathbf K H \in T$. Thus, it suffices to show that
$\mathbf K H \in T^{F_i}_{\rKG}(T_i)$ implies $\mathbf K H \in T$.

By Definition~\ref{def:monotoncoepratos},
\begin{align*}
T^{F_i}_{\rKG}(T_i)
=
\{\, h(r) \mid r\in\mathcal P_G,\;
b^+(r)\subseteq T_i,\;
\mathbf K(b^-(r))\subseteq F_i \,\}
\cup \\
\{\, \mathbf K\xi \mid \mathbf K\xi\in\EKAKG,\;
\OBO{T_i}\models\xi \,\}.
\label{eq:true}
\end{align*}

This definition consists of two parts. We denote the rule part as
\[
T^{F_i}_{\rPG}(T_i)
=
\{\, h(r) \mid r\in\mathcal P_G,\;
b^+(r)\subseteq T_i,\;
\mathbf K(b^-(r))\subseteq F_i \,\},
\]
and the ontology part as
\[
T^{F_i}_{\rO}(T_i)
=
\{\, \mathbf K\xi \mid \mathbf K\xi\in\EKAKG,\;
\OBO{T_i}\models\xi \,\}.
\]
We consider the two cases separately.

\smallskip
\noindent\emph{Rule case.}
Let $\mathbf K H \in T^{F_i}_{\rPG}(T_i)$. Then there exists a rule
$r \in \rPG$ of the form
\[
h(r) \gets b^+(r), b^-(r),
\]
where
\[
b^+(r)=\{\mathbf K A_1,\dots,\mathbf K A_n\}
\quad\text{and}\quad
b^-(r)=\{\textbf{not }\mathbf B_1,\dots,\textbf{not }\mathbf B_m\}.
\]

By definition,
$\mathbf K H \in T^{F_i}_{\rPG}(T_i)$ if and only if $h(r)=\mathbf K H$ and
\[
\{\mathbf K A_1,\dots,\mathbf K A_n\}\subseteq T_i
\quad\text{and}\quad
\{\mathbf K B_1,\dots,\mathbf K B_m\}\subseteq F_i.
\]

By the induction hypothesis IH1,  $\mathbf K A_k\in T$ for all $k$ (\(1 \leq k \leq n\)) and
$\mathbf K B_l\in F$ for all $l$ (\(1 \leq l \leq m\))  . Hence, the body of $r$ is satisfied with
respect to $(T,F)$, and therefore $\mathbf K H\in T$. 
The case $\mathbf K\neg H$ is analogous.

\smallskip
\noindent\emph{Ontology case.}
Let $\mathbf K H \in T^{F_i}_{\rO}(T_i)$. By definition,
\[
\OBO{T_i}\models H,
\]
or respectively $\OBO{T_i}\models \neg H$ in the case of $\mathbf K\neg H$.

By the induction hypothesis IH1, every modal atom occurring in $T_i$ is included in
$T$  and  $\OBO{T_i}\models H$. Hence follow  that $\mathbf K H\in T$.  Since $T_i$ is included in all stable partition and \(\OBO{T_i}\models H \) states that \(\mathbf{K}H\) is also included in $T$ The case $\mathbf K\neg H$ is analogous.
\smallskip
\noindent\emph{(B) Atoms in $F_{i+1}$.}
Recall that $F_{i+1}=F^{(T_i,F_i)}_{\mathcal K_G}$ is the union of all unfounded
sets with respect to $(T_i,F_i)$. Let $\mathbf K H \in F_{i+1}$ (respectively,
$\mathbf K\neg H \in F_{i+1}$). We show that $\mathbf K H$ (respectively,
$\mathbf K\neg H$) belongs to $F$.

From IH1, every modal atom in $T_i$ is included in $T$, and every modal atom in
$F_i$ is included in $F$.

\smallskip
\noindent\emph{Case~1 (Coherency principle).}
Assume that $\OBO{T_i} \models \overline{H}$, where $\overline{A} = \neg A$ and
$\overline{\neg A} = A$.  
Moreover, by IH1, every $\mathbf{K}H \in T_i$ also belongs to $T$.
Hence, if $\OBO{T_i} \models \overline{H}$, then $\mathbf{K}H \in F_{i+1}$.
Consequently, $\mathbf{K}H$ (respectively, $\mathbf{K}\neg H$) belongs to $F$.

\smallskip
\noindent\emph{Case~2}
 $\mathbf K H$ (respectively, $\mathbf K\neg H$) belongs to some
unfounded set $UF$ with respect to $(T_i,F_i)$.
for every   a set of rules $R\subseteq \mathcal P_G$ such that $\mathrm{head}(R) \cup \OBO{T} \models H$, and for all $R' \subset R$, $\mathrm{head}(R') \cup \OBO{T_i} \not\models H$, and for each $\mathbf{K}\xi \in F_i$, the set
        $\mathrm{head}(R) \cup \OBO{T_i} \cup \{\overline{\xi}\}$ is satisfiable;
 there exists a rule $r \in R$ of the form $h(r) \gets b^+(r), b^-(r)$ such that at least one of the following holds:
Moreover, for every such set $R$ there exists a rule $r\in R$ for which at least
one of the following conditions holds:
\begin{enumerate}[leftmargin=*, itemsep=.3pt]
  \item $body^+(r)\cap F_i\neq\emptyset$;
  \item $body^-(r)\cap T_i\neq\emptyset$;
  \item $body^+(r)\cap UF\neq\emptyset$.
\end{enumerate}

We show that in each case there exist a $r \in R$  where $b(r)$ is not satisfied with respect to the
stable partition $(T,F)$, contradicting the assumed support.

If (1) holds, take $\mathbf K A\in body^+(r)\cap F_i$.
By IH1, $\mathbf K A\in F$, and hence the positive body of $r$ is not satisfied
with respect to $(T,F)$.

If (2) holds, take $\mathbf K B\in body^-(r)\cap T_i$.
By IH1, $\mathbf K B\in T$, and therefore the negative body literal of $r$ is
violated with respect to $(T,F)$.

If (3) holds, take $\mathbf K A\in body^+(r)\cap UF$.
Since $UF\subseteq F_{i+1}$ and $F_{i+1}\subseteq F$, it follows that
$\mathbf K A\in F$. Hence, the positive body of $r$ is not satisfied with respect
to $(T,F)$.
In all cases, the assumed set $R$ is required to entail $\mathbf K H$ from the
objective knowledge $\OBO{T_i} \cup \mathit{head}(R)$ (respectively,
$\mathbf K\neg H$). Moreover, for each such assumed set $R$, there exists at least
one rule $r \in R$ whose body is false with respect to $(T_i,F_i)$. By the
induction hypothesis IH1, $(T_i,F_i)$ is included in every stable partition
$(T,F)$. Since $\mathbf K H$ is declared false with respect to $(T_i,F_i)$, it
follows that $\mathbf K H$ (respectively, $\mathbf K\neg H$) belongs to $F$.
\end{proof}

\begin{proposition}\label{prop:trueandfalseineverythreevaluedunita2}
Let \(\rKOPG\) be a ground \(\hybridmknf\) knowledge base and
\((T_E,F_E)\) a fixpoint of \(E_{\rKG}\).
For any modal atom \(\mathbf{K}H \in \EKAKG\), with \(H\) of the form \(\xi\) or \(\neg\xi\),
and for every stable partition \((T,F)\) of \(\rKG\), the following holds:
\(\mathbf{K}H \in T_E\) implies \(\mathbf{K}H \in T\), and
\(\mathbf{K}H \in F_E\) implies \(\mathbf{K}H \in F\).
\end{proposition}
\begin{proof}
Let $E^{(T_\omega,F_\omega)}_{\rKG}\uparrow 0=(\emptyset,\emptyset)$ and
$E^{(T_\omega,F_\omega)}_{\rKG}\uparrow (i+1)
=
E^{(T_\omega,F_\omega)}_{\rKG}(E^{(T_\omega,F_\omega)}_{\rKG}\uparrow i)$
for all $i\ge 0$, and write
$E^{(T_\omega,F_\omega)}_{\rKG}\uparrow i=(T_i,F_i)$.

Fix an arbitrary stable partition of $\rKOPG$ and denote it by $(T,F)$.

We prove by induction on $i\ge 0$ the following property $(\mathcal I_i)$:
for every ground modal atom $\mathbf K H \in \EKAKG$,
if $\mathbf K H \in T_i$ then $\mathbf K H \in T$, and
if $\mathbf K H \in F_i$ then $\mathbf K H \in F$.

\paragraph{Base case ($i=0$).}
Since $E^{(T_\omega,F_\omega)}_{\rKG}\uparrow 0=(\emptyset,\emptyset)$,
property $(\mathcal I_0)$ holds trivially.

\paragraph{Induction step.}
Assume that $(\mathcal I_i)$ holds for
$E^{(T_\omega,F_\omega)}_{\rKG}\uparrow i=(T_i,F_i)$.
We show that $(\mathcal I_{i+1})$ holds for
\[
(T_{i+1},F_{i+1})
=
E^{(T_\omega,F_\omega)}_{\rKG}(T_i,F_i).
\]

\medskip
\noindent\emph{Induction hypothesis (IH).}
For every $\mathbf K H \in \EKAKG$,
if $\mathbf K H \in T_i$ then $\mathbf K H \in T$,
and if $\mathbf K H \in F_i$ then $\mathbf K H \in F$.

By Definition~\ref{def:E-operator},
\[
(T_{i+1},F_{i+1})
=
\bigl(
UPT^{(T_\omega,F_\omega)}_{\rKG}(T_i,F_i)\ \cup\ T^{F_i}_{\rKG}(T_i),\;
UPF^{(T_\omega,F_\omega)}_{\rKG}(T_i,F_i)\ \cup\ F_{\rKG}(T_i,F_i)
\bigr).
\]

We consider atoms added to $T_{i+1}$ and $F_{i+1}$ separately.

\smallskip
\noindent\emph{(A) Atoms in $T_{i+1}$.}
Let $\mathbf K H \in T_{i+1}$.
Then either $\mathbf K H \in T^{F_i}_{\rKG}(T_i)$ or
$\mathbf K H \in UPT^{(T_\omega,F_\omega)}_{\rKG}(T_i,F_i)$.

\smallskip
\noindent\emph{Ontology and rule case.}
If $\mathbf K H \in T^{F_i}_{\rKG}(T_i)$, then the argument is identical to the
corresponding case for the well-founded operator in Proposition \ref{prop:trueandfalseineverythreevaluedwellfoudnedoperaptr}.

\smallskip
\noindent\emph{Unit-propagation case.}
Let $\mathbf K H \in UPT^{(T_\omega,F_\omega)}_{\rKG}(T_i,F_i)$.
By Definition~\ref{def:UP-operators}, there exists a rule
$r\in \rPG$ such that:
\begin{itemize}[leftmargin=*]
  \item $\OBO{T_\omega\cup T_i}\models \overline{h(r)}$,
  \item $b^+(r)\subseteq T_\omega\cup T_i$,
  \item $\mathbf K(b^-(r))\cap(T_\omega\cup T_i)=\emptyset$,
  \item $\mathbf K(b^-(r))\setminus(F_\omega\cup F_i)=\{\mathbf K H\}$.
\end{itemize}

By IH, all modal atoms in $T_i$ (resp.\ $F_i$) are contained in $T$ (resp.\ $F$). 
In-order to satisfy the body of $r$ is satisfied with respect to $(T,F)$ \(\mathbf{K}H\) must be true since \(\mathbf{K}H \in b^-(r)\) 
it follows that $\mathbf K H\in T$.

\smallskip
\noindent\emph{(B) Atoms in $F_{i+1}$.}
Let $\mathbf K H \in F_{i+1}$.
Then either $\mathbf K H \in F_{\rKG}(T_i,F_i)$ or
$\mathbf K H \in UPF^{(T_\omega,F_\omega)}_{\rKG}(T_i,F_i)$.

\smallskip
\noindent\emph{Ontology and unfoundedness case.}
If $\mathbf K H \in F_{\rKG}(T_i,F_i)$, then the argument follows from Proposition \ref{prop:trueandfalseineverythreevaluedwellfoudnedoperaptr}.

\smallskip
\noindent\emph{Unit-propagation case.}
If $\mathbf K H \in UPF^{(T_\omega,F_\omega)}_{\rKG}(T_i,F_i)$, then by
Definition~\ref{def:UP-operators} there exists a rule
$r\in \rPG$ such that:
\begin{itemize}[leftmargin=*]
  \item $\OBO{T_\omega\cup T_i}\models \overline{h(r)}$,
  \item $ K(b^-(r)) \subseteq F_\omega\cup F_i$,
  \item $ b^+(r)\cap(F_\omega\cup F_i)=\emptyset$,
  \item $b^+(r)\setminus(F_\omega\cup F_i)=\{\mathbf K H\}$.
\end{itemize}
By IH, all modal atoms in $T_i$ (resp.\ $F_i$) are contained in $T$ (resp.\ $F$). 
In-order to satisfy the body of $r$ is satisfied with respect to $(T,F)$ \(\mathbf{K}H\) must be false since \(\mathbf{K}H \in b^+(r)\) 
it follows that $\mathbf K H\in F$.
\smallskip
\noindent\emph{Conclusion.}
In all cases, atoms added to $T_{i+1}$ (resp.\ $F_{i+1}$) belong to $T$ (resp.\ $F$).
Thus, $(\mathcal I_{i+1})$ holds.

By induction, the claim holds for all $i\ge 0$, and in particular for the least
fixpoint $(T_E,F_E)=E^{(T_\omega,F_\omega)}_{\rKG}\uparrow\omega$.
\end{proof}

\begin{proposition}
\label{prop:well-foundedmodela1}
Let $\rKOPG$ be a ground $\hybridmknf$ knowledge base, and let
$\mathsf{Out}(\rKG)$ denote the output of
Algorithms~\ref{alg:wfm-stable}, \ref{alg:wfm-stable2}, and~\ref{alg:wfm-stable1}.
(1) If $\mathsf{Out}(\rKG)=\bot$, then $\rKG$ is MKNF-inconsistent. (2) If $\mathsf{Out}(\rKG)=(T_W,F_W)$, then $(T_W,F_W)$ is the well-founded partition
of $\rKG$, and the induced three-valued MKNF interpretation $(M_W,N_W)$ is the
well-founded MKNF model of $\rKG$. (3) If $\mathsf{Out}(\rKG)=\textnormal{\texttt{NoWFM}}$, then $\rKG$ has no well-founded model.
\end{proposition}

\begin{proof}
   \emph{Result (1)}. Let $(T_i,F_i)$ be a partial partition produced at some iteration of
$W_{\rKG}$ or $E^{(T_\omega,F_\omega)}_{\rKG}$.
By Propositions~\ref{prop:trueandfalseineverythreevaluedwellfoudnedoperaptr}
and~\ref{prop:trueandfalseineverythreevaluedunit}, $(T_i,F_i)$ is contained in
every stable partition of $\rKG$.
If $T_i\cap F_i\neq\emptyset$, then some modal atom is both true and false, which
is impossible in any three-valued MKNF model; hence, $\rKG$ is MKNF-inconsistent.
If $\OBO{T_i}$ is unsatisfiable, then no MKNF interpretation pair induced via $(T_i,F_i)$  
(Proposition~\ref{propos:stablepartitionstheorem}), and consequently no
three-valued MKNF model of $\rKG$ exists.
Therefore, if $\mathsf{Out}(\rKG)=\bot$, then $\rKG$ is MKNF-inconsistent.

\emph{Result (1)}. Let $(T,F)$ be an arbitrary stable partition of $\rKG$. By
Proposition~\ref{prop:trueandfalseineverythreevaluedwellfoudnedoperaptr},  Proposition~\ref{prop:trueandfalseineverythreevaluedunit}, and Algorithm \ref{alg:wfm-stable1}, for every modal atom
$\mathbf{K}H \in \EKAKG$, if $\mathbf{K}H \in T_W$ then $\mathbf{K}H \in T$, and if
$\mathbf{K}H \in F_W$ then $\mathbf{K}H \in F$. Hence $T_W \subseteq T$ and
$F_W \subseteq F$. Since $(T_W,F_W)$ is itself a stable partition from Proposition \ref{propos:computingstablepartitions}, it follows by
Definition~\ref{def:well-foundedpartition} that $(T_W,F_W)$ is the well-founded
partition and its induced three-valued MKNF interpretation \((M_W,N_W)\) is the well-founded model from Theorem \ref{theorem:wfmodel}.
\end{proof}

\begin{proposition}\label{props:complexitya1}
Let $\rKOPG$ be a DL-safe $\hybridmknf$ knowledge base. Assume data complexity $\mathcal{C}_1$ for computing unfounded sets with respect to  a partial partition $(T,F)$, and $\mathcal{C}_2$ for DL satisfiability and ground entailment.  Let $(T_W^1,F_W^1)$ and $(T_W^2,F_W^2)$ be the partitions obtained after Phase~2 and Phase~3, respectively. $(i)$ If $(T_W^1,F_W^1)$ coincides with the well-founded partition of $\rKG$, then it can be computed with data complexity $\mathbf{PTime}^{\mathcal{C}_1 \cup \mathcal{C}_2}$.
 $(ii)$ If $(T_W^2,F_W^2)$ coincides with the well-founded partition of $\rKG$ then  it can be computed with data complexity is $\mathbf{EXPTime}^{\mathbf{P}^{\mathcal{C}_1}}$.

\end{proposition}

 \begin{proof}
    \emph{(i)} If $(T_W^1,F_W^1)$ coincides with the well-founded partition of $\rKG$, then it can be computed with data complexity $\mathbf{PTime}^{\mathcal{C}_1 \cup \mathcal{C}_2}$.

Suppose $(T_W^1,F_W^1)=W_{\rKG}\uparrow\omega$.

Here,
\[
W_{\rKG}(T,F)=\bigl(T^F_{\rKG}(T),\,F_{\rKG}(T,F)\bigr).
\]
The computation of $F_{\rKG}(T,F)$ (unfounded sets) has data complexity $\mathcal{C}_1$. The computation of $T^F_{\rKG}(T)$ depends on the data complexity $\mathcal{C}_2$ for DL satisfiability and ground entailment. The operator $W_{\rKG}$ is applied over the finite set $\EKAKG$, hence $W_{\rKG}\uparrow\omega$ reaches a fixpoint in at most $2|\EKAKG|$ iterations. Therefore, $(T_W^1,F_W^1)=W_{\rKG}\uparrow\omega$ can be computed by iterating $W_{\rKG}$ a polynomial number of times, where each iteration depends on computations in $\mathcal{C}_1$ and $\mathcal{C}_2$.

Suppose $(T_W^1,F_W^1)=E_{\rKG}\uparrow\omega$, where
\[
E^{(T_\omega,F_\omega)}_{\rKG}(X,Y)
=
\bigl(
UPT^{(T_\omega,F_\omega)}_{\rKG}(X,Y)\cup T^Y_{\rKG}(X),\,
UPF^{(T_\omega,F_\omega)}_{\rKG}(X,Y)\cup F_{\rKG}(X,Y)
\bigr).
\]
The computations of $T^F_{\rKG}(T)$ and $F_{\rKG}(T,F)$ are as above, as are the computations of \\
$UPT^{(T_\omega,F_\omega)}_{\rKG}(X,Y)$ and $UPF^{(T_\omega,F_\omega)}_{\rKG}(X,Y)$.
In the worst case, $E^{(T_\omega,F_\omega)}_{\rKG}(X,Y)\uparrow\omega$ reaches a fixpoint in at most $|\EKAKG|$ iterations.

$(ii)$ If $(T_W^2,F_W^2)$ coincides with the well-founded partition of $\rKG$ then  it can be computed with data complexity is $\mathbf{EXPTime}^{\mathbf{P}^{\mathcal{C}_1}}$.
Let $(T,F)$ be the partial partition obtained after Phase~2, and let
\[
U \;=\; \EKAKG \setminus (T\cup F)
\]
be the set of atoms whose truth value is still undetermined. Any completion of $(T,F)$ into a total partition over $\EKAKG$ assigns to each atom in $U$ one of the three statuses: true, false, or undefined. Hence, the number of completions is bounded by $3^{|U|}$.

Phase~3 enumerates all such completions, producing at most $3^{|U|}$ candidates $(T',F')$ extending $(T,F)$. For each candidate $(T',F')$, the algorithm checks whether $(T',F')$ is a stable partition of $\rKG$. This check can be carried out in polynomial time in the data, with oracle access to unfounded set computation of data complexity $\mathcal{C}_1$ \cite{LIU2017123}. Therefore, the per-candidate verification runs in $\mathbf{P}^{\mathcal{C}_1}$.

Consequently, the overall computation time is bounded by
\[
3^{|U|} 
\]
with oracle access to $\mathcal{C}_1$, and thus belongs to
$\mathbf{EXPTime}^{\mathbf{P}^{\mathcal{C}_1}}$ in data complexity.

Finally, once the set of stable partitions has been computed (in the worst case containing up to $3^{|U|}$ candidates), selecting $(T_W^2,F_W^2)$ from this set requires comparing candidates and thus takes at most exponential time in $|U|$. Therefore, the selection step does not increase the overall  data-complexity bound, and $(T_W^2,F_W^2)$ can be computed with data complexity $\mathbf{EXPTime}^{\mathbf{P}^{\mathcal{C}_1}}$.

\end{proof}
\end{document}